\documentclass[11pt]{article}
\usepackage{amsmath,amssymb,amsthm}
\usepackage{newtxtext,newtxmath}
\usepackage{geometry}
\usepackage{authblk}
\usepackage[official,right]{eurosym}
\usepackage{graphicx}
\usepackage{appendix}
\usepackage{warpcol}   
\usepackage{enumerate} 
\usepackage{rotating}
\usepackage{tabularx}
\usepackage{caption}
\usepackage{array}
\usepackage{multirow}  
\usepackage{pifont}   
\usepackage{lscape}
\usepackage{setspace}  
\usepackage{pb-diagram} 
\usepackage{fancybox}
\usepackage[colorlinks=true, citecolor=blue, linkcolor=blue, urlcolor=blue]{hyperref}
\usepackage[style=apa]{biblatex} 
\usepackage{booktabs} 
\usepackage{expdlist}
\usepackage{xcolor}
\usepackage{float}
\usepackage{mathtools}
\usepackage{subfigure}
\theoremstyle{plain}
\newtheorem{theorem}{Theorem}
\newtheorem{lemma}{Lemma}
\theoremstyle{definition}

\newtheorem{remark}{Remark}

\usepackage{tikz}
\usetikzlibrary{shapes.geometric, arrows, fit, positioning, decorations.pathreplacing}

\tikzstyle{block} = [rectangle, minimum width=4cm, minimum height=1cm, text centered, draw=black, fill=blue!20]
\tikzstyle{arrow} = [thick,->,>=stealth]
\title{A minimum-risk and cost-efficient two-sample sequential testing framework for the shifted exponential models with application to precipitation data}
\author[1]{Ashwani Rajput\thanks{ORCID: \href{https://orcid.org/0009-0006-8847-8325}{\texttt{0009-0006-8847-8325}}}}
\author[1]{Neeraj Joshi\thanks{ORCID: \href{https://orcid.org/0000-0001-9017-4255}{\texttt{0000-0001-9017-4255}}}\thanks{\textbf{Corresponding Author:} Neeraj Joshi, Department of Mathematics, Indian Institute of Technology Delhi, Hauz Khas, New Delhi, 110016, India. Email: \href{mailto:njoshi@maths.iitd.ac.in}{\texttt{njoshi@maths.iitd.ac.in}}}}
\affil[1]{Department of Mathematics, Indian Institute of Technology Delhi, New Delhi, India}
\date{\today}
\begin{document}
\maketitle
\noindent\textbf{Abstract:} This paper investigates the problem of comparing the location parameters of two shifted exponential models through a novel double sequential sampling framework. The proposed hypothesis testing procedure is developed by controlling the type I error probability at a preassigned level while minimizing a loss function that incorporates both the type II error probability and the associated sampling cost. The corresponding optimal fixed-sample-size expressions are shown to depend on unknown scale parameters, rendering the desired testing accuracies unattainable in practice under fixed-sample designs. To overcome this difficulty, a double sequential sampling procedure is proposed to test the difference between location parameters when the scale parameters are unknown and unequal. The proposed methodology is shown to possess desirable asymptotic properties, including first-order efficiency, second-order efficiency, and second-order risk efficiency. Extensive simulation studies and a real-data application that involves heavy precipitation episodes at meteorological stations demonstrate the practical effectiveness and applicability of the proposed procedure.
\section{Introduction}\label{Intro}
\noindent Comparison of time-to-event characteristics in two treatment groups or populations is crucial in many applied domains including medicine, environmental sciences, reliability engineering, industrial quality control, etc. For example, in medical experiments, the clinician may be interested in determining whether a new drug is more effective in treating patients than an existing drug. In environmental studies, an experimenter may like to compare the inter-arrival patterns of natural phenomenon, such as precipitation episodes, across different geographical locations. Statistical inference with respect to such investigations requires the use of suitable stochastic models that can capture the minimum guaranteed durations or inter-arrival patterns of events across different populations. In practical scenarios, the exponential distribution is widely utilized to model inter-arrival times due to its analytical tractability and relevance in studies involving constant (or near constant) event occurrence rates. In certain experiments, the event of interest often exhibits a minimum duration before its next occurrence. For example, due to prevailing atmospheric circumstances, successive precipitation episodes in a location may be separated by a baseline minimum inter-arrival period. In such cases a two-parameter or shifted exponential distribution appears to be a more realistic model because it contains a location parameter representing a baseline inter-arrival shift or minimum guaranteed duration. Consequently, statistical techniques for comparing minimum guaranteed durations or location parameters in exponentially distributed populations continue to be highly relevant in the literature. Readers may refer to \textcite{KumarPatel1971}, \textcite{RanganathanKale1979}, \textcite{BayoudKittaneh2016}, and \textcite{KrishnamoorthyXia2017} for an overview of two-sample comparison problems involving shifted exponential models.
\bigskip\\
\indent Inferential techniques for the aforementioned comparison problems are often designed within fixed-sample frameworks, selected to meet the desired estimation and testing accuracies. However, sometimes these prescribed sample sizes are quite large and can result in significant experimental costs, longer observation periods, and inefficient utilization of available resources. In such scenarios, researchers are naturally motivated to adopt optimal alternative methods. This is where sequential sampling designs can help in an efficient manner. Sequential techniques allow data collection one-by-one or in stages, thereby reducing unnecessary observations while achieving desired inferential accuracies. Another difficulty in designing fixed-sample statistical techniques is that the associated optimal sample size under given inferential constraints depends on the unknown nuisance parameter(s), creating a barrier in matching the desired statistical inference. This major practical challenge motivates the use of adaptive sequential strategies that can fine-tune the necessary sample size while updating the estimates of nuisance parameter stepwise. When one-by-one sequential sampling seems difficult, practitioners may adopt the idea of collecting data sequentially in groups or stages for operational convenience. Motivated by these challenges, some researchers have developed sequential and multistage estimation and testing procedures for two-sample comparison problems under shifted exponential models. For a brief review, see \textcite{MukhopadhyayHamdy1984}, \textcite{MukhopadhyayDarmanto1988}, \textcite{MukhopadhyayPadmanabhan1993}, \textcite{IsogaiFutschik2010}, \textcite{IsogaiUno2018}, \textcite{MukhopadhyayBapat2016}, and \textcite{mukhopadhyay_aloufi_2024}. 
\bigskip\\
\indent Recently, \textcite{zhuang_bapat_2022} considered the problem of testing the difference in locations of two shifted exponential models using sequential techniques. Specifically, the authors developed two-stage
and purely sequential procedures to determine the appropriate sample sizes while controlling both type I and type II error probabilities at preassigned levels. Similar kind of testing frameworks can be found in \textcite{mukhopadhyay_zhuang_2019} and \textcite{RajputJoshi2025}, however the practical context is different in these papers. The methodologies proposed in the above-mentioned papers are quite useful in controlling error probabilities while preserving other statistical accuracies. Nevertheless, these papers do not explicitly focus on achieving both error control and cost optimization simultaneously within a unified decision-theoretic framework. This significant literature gap motivates us to re-look at the sequential two-sample exponential testing framework in the present paper.
\bigskip\\
\indent Our contributions in this research are summarized below.
\begin{enumerate}
    \item To combine power performance and economical sampling, we propose a novel loss function which incorporates type II error probability together with the sampling cost. We fix the type I error probability at a predetermined low level while minimizing the expected value of this combined loss function.
    \item Since the optimal sample size arising from minimizing the expected loss (or risk) depends on the unknown scale parameters, the required precision cannot be achieved with a fixed sample statistical technique. Hence, we develop a novel double sequential sampling approach to compare minimum inter-arrival durations arising from two independent shifted exponential populations. This technique helps us to reduce the sampling operations substantially while achieving the other optimality criteria.
    \item We deduce several theoretical optimality criteria (such as first- and second-order efficiency and risk efficiency) for the proposed testing framework and achieve the same while working with simulated datasets. 
    \item As an illustration, we present an extensive case study related to precipitation episodes in two different locations. The proposed real-data illustration highlights the applied flavor of the proposed testing framework.
\end{enumerate}
  
  \indent The subsequent sections of the paper are structured as follows. In Section \ref{sec2}, we describe a real-data scenario related to precipitation episodes to motivate the proposed testing problem. The statistical model and formulation of the hypothesis testing problem are described in Section \ref{sec3}. We establish that this problem cannot be tackled with any fixed sample procedure, and a sequential technique is required to deal with it. A novel double sequential sampling strategy with the associated asymptotic properties is the essence of Section \ref{sec4}. In Section \ref{sec5}, we present an extensive simulation study to assess the precision of the proposed theoretical properties under various parameter configurations. We also present a brief power comparison and sensitivity analysis in Section \ref{sec5}. The precipitation data presented in Section \ref{sec2} is studied in detail in Section \ref{sec6} in the context of our proposed testing problem. Finally, we conclude the paper in Section \ref{sec7} with a brief discussion and future scope of the present research. All theoretical proofs are provided in \hyperref[Appendix]{Appendix}. 

\section{A Motivating Real Data Example on the Heavy Precipitation Episodes}\label{sec2}
To motivate the proposed methodology and illustrate its practical relevance, we consider a real-world data set involving the occurrence of heavy precipitation episodes. Data are obtained from the Global Historical Climatology Network-Daily (GHCN-Daily or GHCNd) archives maintained by the National Centers for Environmental Information (NCEI). The analysis is conducted for the following two meteorological stations located in the Pacific Northwest region of the United States:
\begin{itemize}
    \item Seattle-Tacoma International Airport, Washington, USA (Station ID: USW00024233),  
    \item Portland International Airport, Oregon, USA (Station ID: USW00024229).
\end{itemize}
Figure \ref{location_map} displays the geographical locations of the two stations, both of which are influenced by maritime climatic conditions and exhibit pronounced wet-season precipitation patterns. The study focuses on the wet season that spans October through April for the years 2000-2026. Daily precipitation measurements are recorded in the variable \texttt{PRCP}, expressed in tenths of millimeters. A day is classified as a heavy precipitation day whenever
$\texttt{PRCP} \geq 200,$
corresponding to a precipitation amount of at least \(20\) mm.
\bigskip\\
\indent Heavy precipitation associated with a single meteorological system may persist over several consecutive days. As a result, individual exceedances are grouped into clusters, with each cluster representing heavy precipitation episodes. To reduce short-range temporal dependence among exceedances, a de-clustering procedure is employed in which consecutive heavy precipitation days, or exceedances separated by at most two intervening non-exceedance days, are treated as belonging to the same episode. Thus, two episodes are regarded as independent only when they are separated by at least three non-heavy precipitation days.
\bigskip\\
\indent The primary quantity of interest is the inter-arrival time between successive independent heavy precipitation episodes. For each wet season, the inter-arrival time is defined as the number of days between the end of one episode and the beginning of the next episode. Restricting the analysis to individual wet seasons avoids artificially inflated waiting times arising from the relatively dry summer months. The corresponding station records are publicly available through the NCEI archive at the following link:
\[
{\texttt{https://www.ncei.noaa.gov/data/global-historical-climatology-network-daily/access/}}
\]
\indent These datasets provide a motivating real-world setting for studying the stochastic behavior of inter-arrival times associated with heavy precipitation episodes. In particular, an important practical question is whether the minimum inter-arrival times differ significantly between the two stations. This comparison is of considerable practical importance for understanding the patterns of extreme precipitation events in geographically similar regions. It is naturally desired to perform such a comparison using an optimal statistical hypothesis testing approach. A detailed analysis of this problem using the proposed sequential testing framework is presented in Section \ref{sec6}.

\section{The Problem Formulation}\label{sec3}
\noindent Consider two independent two-parameter exponential random variables $X$ and $Y$ with probability density functions (PDFs), $f(x;\mu_1,\sigma_1)$ and $f(y;\mu_2,\sigma_2)$, respectively, where
\begin{equation}
\label{popn_pdf}
    f(t;\mu,\sigma) = \frac{1}{\sigma}\exp\left\{-\frac{1}{\sigma}(t-\mu)\right\} I(t \geq \mu),\,\ \sigma >0,\,\ \mu >0.
\end{equation}
Here, $I(\cdot)$ is the indicator function and the parameters $\mu$ and $\sigma$ are the unknown location and scale parameters, respectively. Suppose that we take $n_1$ samples from population $X$ and $n_2$ samples from population $Y$, and let $\boldsymbol{n} = (n_1, n_2)$. Then, the maximum likelihood (ML) estimators for $\mu_1$ and $\mu_2$ are given by
\begin{equation}
\label{Min_order_stat}
    X_{n_1(1)} = \min(X_1, X_2, \dots, X_{n_1})
    ~~~\text{and}~~~ 
    Y_{n_2(1)} = \min(Y_1, Y_2, \dots, Y_{n_2}).
\end{equation}
\noindent Also, for $n_1\geq2$ and $n_2\geq2$, uniformly minimum variance unbiased estimators (UMVUEs) of $\sigma_1$ and $\sigma_2$ can be obtained as
\begin{equation}
    U^{(1)}_{n_1} = \frac{1}{n_1-1}\sum\limits_{i=1}^{n_1} \left(X_i-X_{n_1(1)}\right)
~~~\text{and}~~~
    U^{(2)}_{n_2} = \frac{1}{n_2-1}\sum\limits_{i=1}^{n_2} \left(Y_i-Y_{n_2(1)}\right).
\end{equation}
\noindent The PDF of $R_{\boldsymbol{n}}= X_{n_1(1)}-Y_{n_2(1)}- (\mu_1-\mu_2)$ can be obtained as
\begin{equation}
\label{R_pdf_in_general}
    f_{R_{\boldsymbol{n}}}(r) = 
    \begin{cases} 
        \left(\frac{\sigma_1}{n_1}+\frac{\sigma_2}{n_2}\right)^{-1}\exp\left(\frac{n_2r}{\sigma_2}\right), &  r \leq 0, \\
        \left(\frac{\sigma_1}{n_1}+\frac{\sigma_2}{n_2}\right)^{-1}\exp\left(\frac{-n_1r}{\sigma_1}\right), &  r > 0.
    \end{cases}
\end{equation}
\indent We are interested in the following hypothesis testing setup:
\begin{equation}
\label{hypothesis}
    H_0: \mu_1-\mu_2 = \Delta_0 \quad \text{vs} \quad H_1: \mu_1-\mu_2 = \Delta_1 \ (>\Delta_0),
\end{equation}
where $\Delta_0$ and $\Delta_1$ are prespecified. Let us denote the parameter vector as $\boldsymbol{\theta} = (\mu_1, \mu_2, \sigma_1,\sigma_2)$, $\boldsymbol{\mu}=(\mu_1, \mu_2)$ and $\boldsymbol{\sigma}=(\sigma_1,\sigma_2)$. In order to test the hypothesis (\ref{hypothesis}), a logical decision rule can be given as
\begin{equation}
\label{decision_rule}
    \text{Reject $H_0$ if and only if} \,\ X_{n_1(1)}-Y_{n_2(1)} > b,
\end{equation}
where $b$ represents a critical value defined later in \eqref{critical_value_expression}, $X_{n(1)}$ and $Y_{n(1)}$ are as described in \eqref{Min_order_stat}. The probabilities of type I and type II errors associated with the decision rule \eqref{decision_rule} are given by
\begin{equation}
\label{typeI_error_prob}
    \mathbb{P}_{\boldsymbol{\theta}}\left(\text{Type I error}\right)= \mathbb{P}_{\boldsymbol{\theta}}\left(X_{n_1(1)}-Y_{n_2(1)} > b \mid H_0 \right) = \left(\frac{\sigma_1}{n_1}+\frac{\sigma_2}{n_2}\right)^{-1}\frac{\sigma_1}{n_1}\exp\left\{-\frac{n_1(b-\Delta_0)}{\sigma_1} \right\}, 
\end{equation}
\begin{equation}
\label{typeII_error_prob}
    \mathbb{P}_{\boldsymbol{\theta}}\left(\text{Type II error}\right)= \mathbb{P}_{\boldsymbol{\theta}}\left(X_{n_1(1)}-Y_{n_2(1)} < b \mid H_1 \right) = \left(\frac{\sigma_1}{n_1}+\frac{\sigma_2}{n_2}\right)^{-1}\frac{\sigma_2}{n_2}\exp\left\{-\frac{n_2(\Delta_1-b)}{\sigma_2} \right\}. 
\end{equation}
\indent The objective is to control the type I error probability at a prespecified level while simultaneously minimizing the type II error probability, taking into account the associated sampling cost. To achieve this, we introduce the following loss function:
\begin{equation}
\label{loss_function}
        L_{\boldsymbol{n}}(\boldsymbol{\mu},\boldsymbol{c}) = A \left[\mathbb{P}_{\boldsymbol{\theta}}\left(\text{Type II error}\right)\right] + C_{\boldsymbol{n}}(\boldsymbol{\sigma}), 
\end{equation}
where, \( A \) is a fixed known positive constant, $\boldsymbol{c}=(c_1,c_2)$ and \( C_{\boldsymbol{n}}(\boldsymbol{\sigma}) > 0 \) represents the total cost of collecting \( (n_1 + n_2) \) observations, defined as
\begin{equation}
\label{cost_function}
    C_{\boldsymbol{n}}(\boldsymbol{\sigma}) = c_1 n_1 \sigma_1^{-1} + c_2 n_2 \sigma_2^{-1}.
\end{equation}
The form of this cost function depends on the specific problem and is determined by the experimenter. It is natural to assume that the cost per observation varies inversely with its uncertainty; consequently, the cost increases (decreases) as the variance or standard deviation decreases (increases). The choice of the weight parameter \( A \), including its units, should likewise be made by the experimenter in accordance with the context of the problem under consideration. The constant $A$ can be interpreted as the penalty associated with committing a type II error. Consequently, larger values of \(A\) are appropriate in situations where erroneous decisions are costly, while smaller values of $A$ may be preferred when moderate inaccuracies are acceptable or when reducing the sampling cost is of greater importance.
\bigskip\\
\indent The associated fixed sample risk (or expected loss) function is defined as 
\begin{equation}
\label{risk_function}
    \begin{split}
        R_{\boldsymbol{n}}(\boldsymbol{\mu},\boldsymbol{c}) &= E_{\boldsymbol{\theta}}(L_{\boldsymbol{n}}(\boldsymbol{\mu},\boldsymbol{c}))\\
        & =A \left[ \left(\frac{\sigma_1}{n_1}+\frac{\sigma_2}{n_2}\right)^{-1}\frac{\sigma_2}{n_2}\exp\left\{-\frac{n_2(\Delta_1-b)}{\sigma_2} \right\} \right] + c_1 n_1 \sigma_1^{-1} + c_2 n_2 \sigma_2^{-1}.
    \end{split}
\end{equation}
Since the variability in the two populations is governed by their respective scale parameters, it is justified to allocate the sample sizes proportionally. To ensure that the two samples contribute equally to the inference regarding $\mu_1$ and $\mu_2$, we require that the Fisher information contained in the minimum order statistics $X_{n_1(1)}$ and $Y_{n_2(1)}$ about the respective location parameters $\mu_1$ and $\mu_2$ be equal. This condition holds if and only if
\begin{equation}
\label{Optimum_allocation}
    \frac{n_1^2}{\sigma_1^2} = \frac{n_2^2}{\sigma_2^2} \quad \text{or equivalently} \quad \frac{n_1}{n_2} = \frac{\sigma_1}{\sigma_2} = u,
\end{equation}
where $u$ is an unknown positive constant. It may be noted that the same condition also ensures that
\[
\mathbb{E}\,\left(X_{n_1(1)} - Y_{n_2(1)}\right) = \mu_1 - \mu_2.
\]
\noindent We now derive the appropriate value of $b$ while ensuring that
\begin{equation}
    \Delta_0 \leq b \leq \Delta_1.
\end{equation}
In view of \eqref{Optimum_allocation}, the probabilities of type I and type II errors can be expressed as
\begin{equation}
\label{typeI_error_prob_n1}
    \mathbb{P}_{\boldsymbol{\theta}}\left(\text{Type I error}\right)= \mathbb{P}_{\boldsymbol{\theta}}\left(X_{n_1(1)}-Y_{n_2(1)} > b \mid H_0 \right) = \frac{1}{2}\exp\left\{-\frac{n_1(b-\Delta_0)}{\sigma_1} \right\}, 
\end{equation}
\begin{equation}
\label{typeII_error_prob_n_1}
    \mathbb{P}_{\boldsymbol{\theta}}\left(\text{Type II error}\right)= \mathbb{P}_{\boldsymbol{\theta}}\left(X_{n_1(1)}-Y_{n_2(1)} < b \mid H_1 \right) = \frac{1}{2}\exp\left\{-\frac{n_1(\Delta_1-b)}{\sigma_1} \right\}. 
\end{equation}
Similarly, the cost function can be expressed as
\begin{equation}
\label{cost_function_n1}
    C_{n_1}(\sigma_1) = (c_1+ c_2) n_1 \sigma_1^{-1}.
\end{equation}
We require that the type I error probability be controlled at a 
prespecified level $\alpha$, i.e.,
    \[
    \frac{1}{2}\exp\left\{-\frac{n_1(b-\Delta_0)}{\sigma_1} \right\} =(\leq)\ \alpha,
    \]
which provides a reasonable choice for $b$ as 
\begin{equation}
\label{b_exp_1}
    b = \Delta_0 - \frac{\sigma_1}{n_1} \log(2\alpha).
\end{equation}
Thus, the type II error probability can be written as 
\[
\mathbb{P}_{\boldsymbol{\theta}}\left(\text{Type II error}\right) = \frac{1}{2}\exp\left\{-\frac{n_1}{\sigma_1}(\Delta_1-\Delta_0) - \log(2\alpha) \right\},
\]
and the associated fixed sample risk function can be defined as 
\begin{equation}
\label{risk_function_n1}
    \begin{split}
        R_{n_1}(\boldsymbol{\mu},\boldsymbol{c}) &= E_{\boldsymbol{\theta}}(L_{n_1}(\boldsymbol{\mu},\boldsymbol{c}))\\
        & =A \left[ \frac{1}{2} \exp \left\{ -\frac{n_1}{\sigma_1}(\Delta_1-\Delta_0) - \log(2\alpha) \right\}  \right] + (c_1+c_2)n_1\sigma_1^{-1}.
    \end{split}
\end{equation}
\indent To simplify optimization, we treat $n_1$ as a continuous variable and differentiate the risk function in \eqref{risk_function}, which is convex in $n_1$, with respect to $n_1$. Setting the derivative equal to zero, we obtain
\[
-\frac{A}{2} \exp \left\{ -\frac{n_1}{\sigma_1}(\Delta_1-\Delta_0) - \log(2\alpha) \right\} \left(\frac{\Delta_1-\Delta_0}{\sigma_1}\right) + (c_1+c_2) \sigma_1^{-1} = 0,
\]
which gives us the expression of optimal fixed sample size $n_1^*$ as follows.
\begin{equation}
\label{optimal_size_1}
    n_1 \equiv n_1^*= \frac{\sigma_1}{(\Delta_1-\Delta_0)} \log \left\{\frac{A(\Delta_1-\Delta_0)}{4\alpha(c_1+c_2)}\right\}.
\end{equation}
Putting this value of $n_1$ in \eqref{b_exp_1}, we get
\begin{equation}
\label{critical_value_expression}
    b = \Delta_0 - \frac{(\Delta_1-\Delta_0) \log (2\alpha)}{\log \left\{\frac{A(\Delta_1-\Delta_0)}{4\alpha(c_1+c_2)}\right\}}.
\end{equation}
Thus, under the optimal allocation rule given in \eqref{Optimum_allocation}, the optimal sample size \( n_2^* \) can be obtained as
\begin{equation}
\label{optimal_size_2}
    n_2 \equiv n_2^*  = \frac{\sigma_2}{(\Delta_1-\Delta_0)} \log \left\{\frac{A(\Delta_1-\Delta_0)}{4\alpha(c_1+c_2)}\right\},
\end{equation}
and the corresponding minimum risk is
\begin{equation}
    R_{\boldsymbol{n^*}}(\boldsymbol{\mu},\boldsymbol{c}) = \left(\frac{c_1+c_2}{\Delta_1-\Delta_0}\right)\left[1+\log\left\{\frac{A(\Delta_1-\Delta_0)}{4\alpha(c_1+c_2)}\right\}\right],
\end{equation}
where, $\boldsymbol{n^*}=(n_1^*,n_2^*)$.
\bigskip\\
\indent Based on \eqref{optimal_size_1} and \eqref{optimal_size_2}, it follows that the optimal fixed sample sizes are given by $n_1^{*}$ and $n_2^{*}$ when $\sigma_1$ and $\sigma_2$ are known. However, in practice, both $\sigma_1$ and $\sigma_2$ are unknown, therefore optimal sample sizes $n_1^{*}$ and $n_2^{*}$ cannot be determined using any fixed sample procedure. Therefore, it becomes necessary to adopt an appropriate sequential or multistage sampling strategy to address this problem. Such procedures allow the sample sizes to be adjusted according to the observed variability, thereby ensuring uniform control over the estimation error.
\bigskip\\
\indent A wide range of sequential and multistage sampling procedures have been proposed in the literature, each characterized by distinct methodologies, along with their respective advantages and limitations. A comprehensive discussion of these procedures can be found in \textcite{mukhopadhyay_2009}. In the present study, we employ a double sequential sampling procedure to estimate the required sample sizes and achieve the desired statistical accuracy.

\section{A Double Sequential Sampling Strategy}
\label{sec4}
Following \textcite{hu2022double}, we propose the following methodology to collect the data: we begin with initial samples of sizes $m(\geq 2)$, denoted by $X_1, \ldots, X_{m}$ and $Y_1, \ldots, Y_{m}$, drawn from populations $X$ and $Y$, respectively, as defined in \eqref{popn_pdf}. Furthermore, we consider
\begin{equation}
    d_\alpha = \frac{1}{(\Delta_1-\Delta_0)} \log \left\{\frac{A(\Delta_1-\Delta_0)}{4\alpha(c_1+c_2)}\right\}.
\end{equation}
Then, for $i = 1,2$, we fix the constants $\rho_i \in (0,1)$ and integers $k_i(\geq 2)$. The parameters $\rho_1$ and $\rho_2$, as well as $k_1$ and $k_2$ can be taken to be equal if desired. Subsequently, $k_i$ additional observations are collected at each stage sequentially until the following stopping criterion is met for the first time
\begin{equation}
    \label{Stage_1}
    L_i\equiv L_i(k_i,\rho_i) = \inf \left\{ n \geq 0 : m + k_in \geq \rho_i\ U^{(i)}_{m+k_in}d_{\alpha} \right\}.
\end{equation}
The total sample size at the stopping point is 
\[
T_i\equiv T_i(k_i,\rho_i) = m + k_iL_i.
\]
We then proceed by collecting observations sequentially, one at a time, until the following stopping rule \eqref{Stage_2} is satisfied for the first time
\begin{equation}
    \label{Stage_2}
    N_i\equiv N_i(k_i,\rho_i) = \inf \left\{ n \geq T_i : n \geq U^{(i)}_{n}d_{\alpha} \right\}.
\end{equation}
The sampling procedure presumes that the observations obtained from each of the two populations are collected independently of each other. Let $\boldsymbol{N} = (N_1, N_2)$ denote the vector of sample sizes. It follows that $\mathbb{P}(N_i < \infty) = 1$ as $q = \log \left(\frac{1}{c_1+c_2}\right) \to \infty$ (see \cite{chow_robbins_1965}). Upon termination of the procedure, we obtain the final samples $\{X_1, \ldots, X_{N_1};\, Y_1, \ldots, Y_{N_2}\}$. Based on these final samples, we implement the following decision rule for hypothesis testing, as defined in \eqref{decision_rule}
\begin{equation}
\label{decision_rule_2}
    \text{Reject $H_0$ if and only if} \,\ X_{N_1(1)}-Y_{N_2(1)} > b,
\end{equation}
where $b$ is specified in \eqref{critical_value_expression}.
\bigskip\\
\indent Note that by taking $k_i = 1$ for $i = 1,2$ in the proposed procedure, we obtain the corresponding purely sequential sampling scheme, given by  
\begin{equation}
    \label{purely_seq_rule}
    N_i' \equiv N_i(1,\rho_i) = \inf \left\{ n \geq m : n \geq U^{(i)}_{n} d_{\alpha} \right\}.
\end{equation}
\indent Furthermore, we assume that the following conditions hold:
\begin{equation}
\label{pilot_sample_assumption}
    m = k_im_{i0} + 1,\,\ \text{where } m_{i0} > \frac{2}{k_i} \text{ is an integer, and }\ \frac{m}{n_i^*} < \rho_i\ .
\end{equation}
It may be noted that, in practice, although it is preferable to choose $m$ and $k_i$ such that $m - 1 \equiv 0 \pmod{k_i}$, this condition is not essential and may be relaxed. Therefore, minor deviations from this requirement are generally not a matter of concern.
\bigskip\\
\indent Before moving directly to the optimal asymptotic properties of the proposed sequential rule, we wish to highlight the need of such a double sequential sampling framework with a brief review of the literature in the following section.

\subsection{Why Double Sequential Sampling?}\label{sec4.2}
\noindent \textcite{stein_1945} introduced a two-stage sampling procedure requiring only two sampling stages. However, in Stein's procedure, the difference between the average sample number and the optimal fixed sample size is not asymptotically bounded, which may lead to substantial oversampling. To address this limitation, \textcite{anscombe_1953}, \textcite{chow_robbins_1965}, and \textcite{starr_1966} (ACR) developed the theory of sequential estimation. In such procedures, observations are collected one at a time, which can be operationally demanding in practice, although these methods possess desirable asymptotic properties. Nevertheless, ACR-type designs typically involve a large number of sampling operations, making them potentially expensive to implement. In an effort to combine the advantages of two-stage and sequential methodologies, \textcite{hall_1983} proposed an accelerated sequential sampling procedure. This approach reduces the number of sampling stages by a predetermined factor while incurring only a finite increase in the total number of observations. Specifically, the procedure begins with a purely sequential phase to obtain a preliminary estimate, followed by a single-stage collection of the remaining observations in a batch.
\bigskip\\
\indent More recently, in a similar spirit of reducing the sampling effort, \textcite{mukhopadhyay2020groupseq} introduced a group sequential sampling scheme in which observations are collected in batches of size $k$, rather than individually. Although this approach reduces the number of sampling operations, it may result in oversampling since additional observations could be collected at the termination stage. To mitigate this issue, \textcite{hu2022double} proposed a double sequential sampling scheme. In this design, sampling is initially conducted in groups of size $k$ and subsequently observations are collected one-at-a-time. This hybrid approach achieves sample sizes equivalent to purely sequential procedures, while reducing the overall sampling operations. The effectiveness of the proposed double sequential sampling procedure for the comparison problem considered in this paper is demonstrated through Theorems \ref{thm1}-\ref{thm3}, as well as the data analysis presented in Sections \ref{sec5} and \ref{sec6}, respectively.
\bigskip\\
\indent In practical applications, the selection of $k$ and $\rho$ should be carried out carefully to maintain an appropriate balance between operational feasibility and sampling cost. When the cost per observation is high, it is generally preferable to choose a smaller value of $k$ to limit excessive oversampling, together with a moderate or relatively large value of $\rho$, depending on the desired level of operational convenience. In contrast, when both time and operational considerations are of greater importance, larger values of $k$ and $\rho$ may be adopted, as they can substantially reduce the number of sampling operations and improve overall efficiency.

\subsection{Asymptotic Properties}\label{sec4.1}
\noindent We begin by establishing a preliminary result (Lemma \ref{lemma1}), which will be instrumental in deducing the first- and second-order asymptotic properties (Theorems \ref{thm1}--\ref{thm3}) of the procedure described in \eqref{Stage_1}-\eqref{Stage_2}. In particular, these properties include the criterion of efficiency and consistency for the proposed double sequential sampling rule.
\begin{lemma}
\label{lemma1}
    Let $N_i'$ be the stopping rule corresponding to the purely sequential sampling scheme \eqref{purely_seq_rule}, and let $T_i$ be the stopping rule defined in \eqref{Stage_1}. Then, for $s>0$, we have
    \[
    \mathbb{P}_{\boldsymbol{\theta}}(N_i' < T_i) = O(n_i^{*-s/2}),\ \text{ provided that }  m>\frac{s}{2}+1.
    \]
\end{lemma}
\begin{theorem}
\label{thm1}
    For the double sequential sampling scheme defined in \eqref{Stage_1}-\eqref{Stage_2} along with the decision rule defined in \eqref{decision_rule} for the hypothesis testing problem \eqref{hypothesis} with fixed $\Delta_0, \Delta_1, k_i\geq 2, 0<\rho_i< 1$, we have as $q = \log \left(\frac{1}{c_1+c_2}\right) \to \infty$,
    \begin{enumerate}[(i)]
        \item $ \mathbb{E}_{\boldsymbol{\theta}}(N/n^*) \longrightarrow 1\,\ [\text{first-order efficiency}],$
        \item $\mathbb{E}_{\boldsymbol{\theta}}(N_i) = n_i^* + a_2 + o(1), \,\ \text{if } m>5/2\ \ [\text{second-order efficiency}]$,
    \end{enumerate}
    where $a_2 \approx -1/4 $.
\end{theorem}
    
\begin{theorem}
\label{thm2}
    The asymptotic second-order expansion (as $q \to \infty$) of the type I error probability associated with the decision rule \eqref{decision_rule} for the hypothesis testing problem \eqref{hypothesis}, under the same parameter settings as in Theorem \ref{thm1}, is given by
    \[
    \mathbb{P}_{\boldsymbol{\theta,N}}(\text{Type I error})
    = \mathbb{E}_{\boldsymbol{\theta}}\left[
    \left(\frac{\sigma_1}{N_1}+\frac{\sigma_2}{N_2}\right)^{-1}
    \frac{\sigma_1}{N_1}
    \exp\left\{-\frac{N_1(b-\Delta_0)}{\sigma_1} \right\}
    \,\middle|\, H_0
    \right]= \alpha + b_1 + o(q^{-1}),
    \]
    provided that $m \geq 8$, where $b_1\approx \frac{\alpha}{d_\alpha}\left(\frac{3}{8\sigma_1}-\frac{3}{8\sigma_2}+\frac{l_1^2}{2\sigma_1}+\frac{3l_1}{4\sigma_1}-\frac{l_1}{32\sigma_1\sigma_2d_\alpha}\right)$ and $l_1=d_\alpha (b-\Delta_0)=-\log(2\alpha)$.

\end{theorem}

\begin{theorem}
\label{thm3}
   Under the same assumptions as in Theorem~\ref{thm1}, the risk efficiency and regret of the proposed double-sequential sampling procedure are as follows:
\begin{enumerate}[(i)]    
    \item Risk Efficiency: $\frac{R_{\boldsymbol{N}}(\boldsymbol{\mu},\boldsymbol{c})}{R_{\boldsymbol{n^*}}(\boldsymbol{\mu},\boldsymbol{c})} \rightarrow 1$,
    \item Regret: $R_{\boldsymbol{N}}(\boldsymbol{\mu},\boldsymbol{c})-R_{\boldsymbol{n^*}}(\boldsymbol{\mu},\boldsymbol{c}) =b_2-\frac{1}{4}\left(\frac{c_1}{\sigma_1}+\frac{c_2}{\sigma_2}\right) + o(q^{-1})$.
\end{enumerate}
where $m\geq 4$, $b_2\approx \frac{A\exp(-l_2)}{4 d_\alpha}\left[-\frac{3}{4\sigma_1}+\frac{3}{4\sigma_2}+\frac{l_2^2}{\sigma_2}+\frac{3l_2}{2\sigma_2}-\frac{l_2}{16\sigma_1\sigma_2d_\alpha}\right]$ and $l_2=d_\alpha(\Delta_1-b)= \log \left\{\frac{A(\Delta_1-\Delta_0)}{2(c_1+c_2)}\right\}$. 
\end{theorem}
\begin{remark}
Note that the constants $b_1$ and $b_2$ in Theorems~\ref{thm2} and~\ref{thm3}, respectively, tend to zero as $q \to \infty$.
\end{remark}
\begin{remark}
The \textit{regret} defined in part (ii) of Theorem~\ref{thm1} highlights the difference in risk functions of the proposed double sequential rule and the optimal fixed sample size procedure. Under the proposed testing framework, our aim is to achieve smaller or negative regret values. 
\end{remark}
\section{Simulation Study}\label{sec5}
We proceed to validate the theoretical results by means of a comprehensive simulation study. All simulations are conducted using the \textit{R} software environment. The simulation framework is developed separately under the null and alternative hypotheses. Initially, pseudorandom samples are generated from populations $X \sim f(x; \mu_1, \sigma_1)$ and $Y \sim f(y; \mu_2, \sigma_2)$, where $f(\cdot)$ denotes the PDF of the two-parameter exponential distribution, as specified in \eqref{popn_pdf}. The following parameter configurations are considered:
\begin{enumerate}[(i)]
    \item $\mu_1 = 25$, $\mu_2 = 20$, $\sigma_1 = 9$, $\sigma_2 = 12$,
    \item $(\rho_1, \rho_2) \in \{(0.3, 0.4), (0.5, 0.7)\}$ and $k_1 = k_2 \in \{1, 4, 6\}$, with $k = 1$ representing the purely sequential benchmark,
    \item $A=10^3,\ \alpha = 0.05$ and $m = 13$, satisfying $m - 1 \equiv 0 \pmod{k_i}$ for $i = 1,2$, and the conditions of Theorems \ref{thm1}-\ref{thm3},
    \item Under $H_0$, $(\Delta_0, \Delta_1) = (5, 6)$ and $(n_1^*,n_2^*) \in \left\{(45, 60),\ (150, 200),\ (240, 320),\ (450, 600)\right\}$,
    \item Under $H_1$, $(\Delta_0, \Delta_1) = (4, 5)$ and $(n_1^*,n_2^*) \in \left\{(45, 60),\ (120,160),\ (270,360),\ (420,560)\right\}$.
\end{enumerate}
Subsequently, an initial sample of size $m$ is generated from populations $f(x; 25, 9)$ and $f(y; 20, 12)$, respectively. At each stage, additional $k_1$ and $k_2$ observations are drawn independently from $X$ and $Y$, respectively in a sequential manner until the stopping criterion \eqref{Stage_1} is first satisfied. Thereafter, observations are collected one at a time until the stopping rule \eqref{Stage_2} is met. This yields the terminal sample $\{X_1, \ldots, X_{N_1};\, Y_1, \ldots, Y_{N_2}\}$. The procedure is replicated $10{,}000$ times to evaluate the performance of \eqref{Stage_1}--\eqref{Stage_2} and the decision rule \eqref{decision_rule} in terms of first- and second-order asymptotic properties.
\bigskip\\
\indent Although the optimal fixed sample sizes $n_1^{*}$ and $n_2^{*}$ are unknown, they are computed using \eqref{optimal_size_1}-\eqref{optimal_size_2} for comparison with their corresponding estimates $N_1$ and $N_2$, respectively. These values are used to assess the performance of first- and second-order approximations across small, moderate, and large sample sizes. The simulation results under the null and alternative hypotheses are presented in Tables \ref{Sim_Under_H0}-\ref{Sim_Under_H1}, with the reported quantities described in Table \ref{tables_col_details}. Based on these results, several key observations are summarized below.
\begin{enumerate}[(i)]
    \item The results for the proposed double sequential sampling scheme ($k_i \geq 2,\ 0<\rho_i<1$) are comparable to those obtained under the purely sequential procedure ($k_i = 1$).
    \item The estimates $\bar{N}_1$ and $\bar{N}_2$ (Column 5) closely approximate $n_1^*$ and $n_2^*$, respectively, with small standard errors (SEs) across all choices of $k_i$ and $\rho_i$.
    \item The ratios $\bar{N}_i/n_i^*$ (Column 6) approach unity as the optimal sample sizes increase, demonstrating first-order efficiency.
    \item The differences $(\bar{N}_i - n_i^*)$ are consistently negative, indicating slight underestimation. However, the deviations are negligible relative to $n_i^*$, even for large values, in accordance with Theorem \ref{thm1}(ii), confirming second-order efficiency.
    \item Under $H_0$, the estimated type I error probabilities (Column 8, Table~\ref{Sim_Under_H0}) remain close to the nominal level $\alpha = 0.05$, with small SEs; minor inflation is observed for smaller sample sizes.
    \item The risk ratios $\bar{\omega}$ (Column 8, Table~\ref{Sim_Under_H1}) are close to unity with negligible SEs, indicating strong first-order risk efficiency.
    \item The average regret $\bar{\kappa}$ (Column 9, Table~\ref{Sim_Under_H1}) is negligible, except for small sample sizes, consistent with Theorem~\ref{thm3}(ii), confirming second-order risk efficiency.
    \item The number of sampling operations required (Column 9, Table~\ref{Sim_Under_H0}; Column 10, Table~\ref{Sim_Under_H1}) highlights the advantage of the proposed double sequential procedure over the purely sequential approach without loss of efficiency.
\end{enumerate}

\subsection{Power Comparisons and Sensitivity Analysis}
This section presents a brief analysis on the power and sensitivity of the proposed double sequential sampling procedure \eqref{Stage_1}-\eqref{Stage_2}.
\begin{enumerate}
    \item \textbf{Power analysis}
\begin{itemize}
        \item \textit{Parameter configuration:} $\mu_1 = 25$, $\mu_2 \in [18.04,\,22.00]$, $\sigma_1 = 9$, $\sigma_2 = 12$, $\Delta_0 = 5$, $A = 10^3$, $k_1 = k_2 = 4$, $\rho_1 = 0.5$, $\rho_2 = 0.7$, $m = 13$, $c_1 + c_2 = 10^{-5}$, $\alpha = 0.05$.      
        \item \textit{Observations:} Power curves for $\Delta_1 = 6$ and $7$ (Figures~\ref{power_curve_6} and \ref{power_curve_7}) exhibit high power when $\mu_1 - \mu_2$ is close to $\Delta_1$, and decrease monotonically to zero as $\mu_2$ increases.
\end{itemize}
\item \textbf{Sensitivity with respect to $(c_1+c_2)$}
\begin{itemize}
    \item \textit{Parameter configuration:} $\mu_1 = 25$, $\mu_2 = 20$, $\Delta_0 = 5$, $\Delta_1 = 6$, and $(c_1 + c_2) \in [10^{-19.75},\,10^{2.5}]$, with all remaining parameters held fixed.
    \item \textit{Observations:} As shown in Figure~\ref{sensitivity_wrt_cost}, $(N_{1}, N_{2})$ decrease approximately linearly (on a logarithmic scale) as $(c_1+c_2)$ increases. The type I error $\alpha$ remains stable near its nominal level of $0.05$. The statistical power $(1-\beta)$ is close to unity for small $(c_1+c_2)$ (corresponding to larger optimal sample sizes), while it exhibits a slight decrease for larger $(c_1+c_2)$ values. Moreover, substantial differences in total sampling operations across different combinations of $k_i$ and $\rho_i$ are observed for small $(c_1+c_2)$, whereas these differences converge as $(c_1+c_2)$ becomes large.
\end{itemize}
\item \textbf{Sensitivity with respect to $m$}
\begin{itemize}
    \item \textit{Parameter configuration:} $(\Delta_0,\, \Delta_1) = (5,\,6)$ for $(N_{1}, N_{2})$ and $\alpha$, and $(3,\,5)$ for regret; all other parameters remain as in the previous setting. The pilot sample size varies over $m \in [9,\,59]$.  
    \item \textit{Observations:} Figure~\ref{sensitivity_wrt_m} indicates that $(N_{1}, N_{2})$, $\alpha$, and regret exhibit negligible sensitivity to $m$.
\end{itemize}
\end{enumerate}

\section{Illustration Using Precipitation Data}\label{sec6}
After preprocessing and declustering, the resulting datasets defined in Section \ref{sec2} contain \(146\) inter-arrival times for Seattle-Tacoma international airport and \(110\) inter-arrival times for Portland international airport. As a preliminary step, the suitability of the shifted exponential distribution for modeling the inter-arrival times of heavy precipitation episodes is examined using histograms, quantile-quantile (Q-Q) plots, and empirical cumulative distribution function (ECDF) plots for both stations, as presented in Figures \ref{Histogram}-\ref{ECDF}. The histograms display a clear right-skewed pattern with exponentially decaying densities, consistent with the proposed model. In both cases, the fitted shifted exponential curves provide a reasonable approximation to the observed distributions, although a few large waiting times are visible in the upper tails. The Q-Q plots show strong agreement between the empirical and theoretical quantiles, with only minor deviations appearing in the extreme tail regions, particularly for the Portland dataset. Likewise, the ECDFs closely follow the corresponding fitted theoretical distribution functions over most of the support. Only mild discrepancies are observed in the tail regions, while the central portions of the distributions show strong agreement. These graphical findings are further supported by the two-sided Kolmogorov-Smirnov (KS) test, which yields $p$-values of $0.6359$ and $0.2631$ for the Seattle-Tacoma international airport and Portland international airport datasets, respectively. Since both $p$-values exceed the conventional significance level of $0.05$, the shifted exponential distribution appears to provide an adequate model for the inter-arrival times of heavy precipitation episodes at both locations.
\bigskip\\
\indent Let \(X\) and \(Y\) denote the inter-arrival time distributions corresponding to the Seattle--Tacoma international airport and Portland international airport datasets, respectively. To implement the proposed double sequential procedure, the design parameters are chosen as $\Delta_0 = 0,\, \alpha = 0.05,\,  A = 10,$ with $k_1 = k_2 = 3, \, \rho_1 = \rho_2 = 0.6$. The null value \(\Delta_0\) is examined against two alternatives, namely \(\Delta_1 = 4\) and \(\Delta_1 = 6\). For each choice of \(\Delta_1\), two values of the total cost coefficient, $(c_1+c_2)=10^{-2}\ \text{and} \ 10^{-3}$, are considered. The pilot sample size is fixed at \(m=8\). Treating processed rainfall datasets as populations, the corresponding ML estimates of the model parameters are $\hat{\theta}_1=\hat{\theta}_2=4$ and $\hat{\sigma}_1=18.669,\, \hat{\sigma}_2=22.486$. 
\bigskip\\
\indent Following the sampling schemes described in \eqref{Stage_1}--\eqref{Stage_2}, observations are selected using simple random sampling without replacement (SRSWOR), and the final decision is determined according to the rule given in \eqref{decision_rule_2}. To assess the sampling variability and stability of the estimator $(X_{N_1(1)} - Y_{N_2(1)})$, a nonparametric bootstrap procedure is additionally employed. For each case, bootstrap resampling with replacement is performed from terminal samples of sizes \(N_1\) and \(N_2\), respectively. Based on \(B'=1000\) bootstrap replications, bootstrap statistics $\Psi^{*(1)}_{(N_1,N_2)},\ldots,\Psi^{*(B')}_{(N_1,N_2)}$ are obtained, where
\[
\Psi^{*(j)}_{(N_1,N_2)} = \left(X^{*(j)}_{N_1(1)}-Y^{*(j)}_{N_2(1)}\right),\quad j=1,\ldots,B',
\]
denotes the bootstrap replicate of $X_{N_1(1)} - Y_{N_2(1)}$.
The corresponding bootstrap mean and bootstrap SE are denoted by $\overline{\Psi}^{*}_{(N_1,N_2)}$ and $s_{\Psi^{*}_{(N_1,N_2)}}$ respectively.
\bigskip\\
\indent The results are presented in Table \ref{real_data_analysis}. It should be noted that the reported terminal sample sizes and the corresponding values of \(X_{N_1(1)}-Y_{N_2(1)}\) are based on a single realization of the proposed sequential sampling procedure. In all scenarios considered, the estimated sample sizes are reasonably close to the corresponding optimal fixed-sample sizes. Moreover, the observed statistic $(X_{N_1(1)}-Y_{N_2(1)})$ is smaller than the critical value \(b\), leading to a failure to reject the null hypothesis in each case. It may further be noted that the repeated zero values of \(X_{N_1(1)}-Y_{N_2(1)}\) arise from the discrete nature of the inter-arrival time data and the use of the minimum order statistic as the estimator of \(\mu\), due to which the terminal samples from both stations frequently attain the same minimum value. The bootstrap summaries indicate moderate variability in the estimator \(X_{N_1(1)}-Y_{N_2(1)}\), which is expected since the procedure is based on extreme order statistics. Nevertheless, the bootstrap mean estimates remain close to the observed statistic, suggesting reasonable stability of the estimated difference under repeated resampling. Consequently, the analysis does not provide sufficient statistical evidence to conclude that the minimum inter-arrival times of heavy precipitation episodes differ between the two stations.

\section{Concluding Thoughts}\label{sec7}
In classical hypothesis testing, fixing the type I error probability generally leads to a reduction in the type II error probability as the sample size increases. A common approach is therefore to prescribe both error probabilities in advance and determine the required sample size accordingly. However, such a framework may create an unrealistic impression that observations can be collected without cost or that the available sampling budget is effectively unlimited. To address this limitation, the present work adopts a cost-efficient testing framework in which the type I error probability is controlled at a preassigned level while the type II error probability is minimized, equivalently maximizing the power of the test, subject to a given sampling cost constraint. Motivated by this objective, the paper develops double sequential sampling procedures for testing the difference between the location parameters of two shifted exponential distributions with unknown and unequal scale parameters. Since the optimal fixed sample size expressions involve nuisance parameters, the proposed sequential procedure provides a practically implementable alternative. The proposed methodology is computationally efficient and, through suitable choices of the design parameters \((\rho_i, k_i)\), requires fewer sampling operations compared with existing other conventional sequential procedures. Theoretical investigations establish desirable asymptotic properties of the proposed procedures, while simulation studies and real-data analysis demonstrate their practical effectiveness and applicability. Future research may focus on extending the proposed framework to more general distributional settings, multiple-comparison problems, and multi-sample situations under suitable sequential sampling schemes.
\bigskip\\
\noindent\textbf{Conflict of Interest}
\bigskip\\
The authors declare no conflict of interest.
\bigskip\\
\noindent\textbf{Data Availability}
\bigskip\\
The dataset used for the illustrative analysis in this research is publicly accessible via the link provided in Section \ref{sec2}.
\bigskip\\
\noindent\textbf{Funding}
\bigskip\\
Ashwani Rajput greatly acknowledges the financial support received from the Indian Institute of Technology Delhi to pursue this research as part of his doctoral thesis.

\printbibliography

@article{hu2022double,
  title   = {A Double-Sequential Sampling Scheme},
  author  = {Hu, J.},
  journal = {Communications in Statistics - Theory and Methods},
  pages   = {6319--6333},
  volume  = {51},  
  year    = {2020}
}

@article{hall_1983,
  title={Sequential estimation saving sampling operations},
  author={Hall, P.},
  journal={Journal of the Royal Statistical Society Series B: Statistical Methodology},
pages={219--223},
  volume={45},
  year={1983}
}

@article{mukhopadhyay2020groupseq,
  title   = {Purely sequential {FWCI and MRPE} problems for the mean of a normal population by sampling in groups with illustrations using breast cancer data},
  author  = {Mukhopadhyay, N. and Wang, Z.},
  journal = {Sequential Analysis},
  volume  = {39},
  pages   = {176--213},
  year    = {2020}
}

@article{mukhopadhyay_zhuang_2019,
  title={Two-sample two-stage and purely sequential methodologies for tests of hypotheses with applications: comparing normal means when the two variances are unknown and unequal},
  author={Mukhopadhyay, N. and Zhuang, Y.},
  journal={Sequential Analysis},
  pages={69 -- 114},
  volume={38},
  year={2019}
}

@article{woodroofe_1977,
  title={Second-order approximations for sequential point and interval estimation},
  author={Woodroofe, M.},
  journal={Annals of Statistics},
  pages={984 -- 995},
  volume={5},
  year={1977}
}

@article{mukhopadhyay_aloufi_2024,
  title={Second-order (s.o.) multi-stage fixed-width confidence
interval ({FWCI}) estimation strategies for comparing
location parameters from two negative exponential ({NE})
populations: illustrations with cancer data},
  author={Mukhopadhyay, N. and Aloufi, A.},
  journal={Metrika},
  pages={649 -- 680},
  volume={87},
  year={2024}
}

@article{stein_1945,
    title = {A two sample test for a linear hypothesis whose power is independent of the variance},
    author = {Stein, C.},
    journal = {Annals of Mathematical Statistics},
pages={243 -- 258},
  volume={16},
    year = {1945}
}

@article{anscombe_1953,
    title = {Sequential estimation},
    author = {Anscombe, F.J.},
    journal = {Journal of the Royal Statistical Society Series B: Statistical Methodology},
pages={1 -- 29},
  volume={15},
    year = {1953}
}

@article{chow_robbins_1965,
    title = {On the asymptotic theory of fixed width sequential confidence intervals for the mean},
    author = {Chow, Y.S. and Robbins, H.},
    journal = {Annals of Mathematical Statistics},
pages={457 -- 462},
  volume={36},
    year = {1965}
}

@article{starr_1966,
    title = {On the asymptotic efficiency of a sequential procedure for estimating the mean},
    author = {Starr, N.},
    journal = {Annals of Mathematical Statistics},
pages={1173 -- 1185},
  volume={37},
    year = {1966}
}

@book{mukhopadhyay_2009,
    author = {Mukhopadhyay, N and de Silva, BM},
    title = {Sequential Methods and Their Applications},
    publisher = {Boca Raton, FL: CRC},
    year = {2009}
}

@article{RanganathanKale1979,
  author    = {Ranganathan, J. and Kale, B.K.},
  title     = {Tests of Hypotheses for Reliability Functions in Two-Parameter Exponential Models},
  journal   = {Canadian Journal of Statistics},
  volume    = {7},
  number    = {2},
  pages     = {177--184},
  year      = {1979}
}

@article{KumarPatel1971,
  author    = {Kumar, S. and Patel, H.I.},
  title     = {A Test for the Comparison of Two Exponential Distributions},
  journal   = {Technometrics},
  volume    = {13},
  pages     = {183--189},
  year      = {1971}
}

@article{KrishnamoorthyXia2017,
  author    = {Krishnamoorthy, K. and Xia, Y.},
  title     = {Confidence Intervals for a Two-Parameter Exponential Distribution: One and Two-Sample Problems},
  journal   = {Communications in Statistics -- Theory and Methods},
  volume    = {47},
  pages     = {935--952},
  year      = {2017}
}

@article{BayoudKittaneh2016,
  author    = {Bayoud, H.A. and Kittaneh, O.A.},
  title     = {Testing the Equality of Two Exponential Distributions},
  journal   = {Communications in Statistics -- Simulation and Computation},
  volume    = {45},
  pages     = {2249--2256},
  year      = {2016}
}

@article{MukhopadhyayHamdy1984,
  author  = {Mukhopadhyay, N. and Hamdy, H.I.},
  title   = {On Estimating the Difference of Location Parameters of Two Negative Exponential Distributions},
  journal = {Canadian Journal of Statistics},
  volume  = {12},
  number  = {1},
  pages   = {67--76},
  year    = {1984}
}

@article{IsogaiFutschik2010,
  author  = {Isogai, E. and Futschik, A.},
  title   = {Sequential Estimation of a Linear Function of Location Parameters of Two Negative Exponential Distributions},
  journal = {Journal of Statistical Planning and Inference},
  volume  = {140},
  number  = {9},
  pages   = {2416--2424},
  year    = {2010}
}

@article{IsogaiUno2018,
  author  = {Isogai, E. and Uno, C.},
  title   = {Three-Stage Confidence Intervals for a Linear Combination of Locations of Two Negative Exponential Distributions},
  journal = {Metrika},
  volume  = {81},
  number  = {1},
  pages   = {85--103},
  year    = {2018}
}

@article{MukhopadhyayBapat2016,
  author  = {Mukhopadhyay, N. and Bapat, S.R.},
  title   = {Multistage Estimation of the Difference of Locations of Two Negative Exponential Populations under a Modified Linex Loss Function: Real Data Illustrations from Cancer Studies and Reliability Analysis},
  journal = {Sequential Analysis},
  volume  = {35},
  number  = {3},
  pages   = {387--412},
  year    = {2016}
}

@article{MukhopadhyayDarmanto1988,
  author  = {Mukhopadhyay, N. and Darmanto, S.},
  title   = {Sequential Estimation of the Difference of Means of Two Negative Exponential Populations},
  journal = {Sequential Analysis},
  volume  = {7},
  number  = {2},
  pages   = {165--176},
  year    = {1988}
}

@article{MukhopadhyayPadmanabhan1993,
  author  = {Mukhopadhyay, N. and Padmanabhan, A.R.},
  title   = {A Note on Three-Stage Confidence Intervals for the Difference of Locations: The Exponential Case},
  journal = {Metrika},
  volume  = {40},
  number  = {1},
  pages   = {121--128},
  year    = {1993}
}

@article{zhuang_bapat_2022,
  title={On comparing locations of two-parameter exponential distributions using sequential sampling with applications in cancer research},
  author={Zhuang, Y. and Bapat, S.R.},
  journal={Communications in Statistics - Simulation and Computation},
  pages={6114 -- 6135},
  volume={51},
  year={2022}
}

@article{RajputJoshi2025,
  author  = {Rajput, A. and Joshi, N.},
  title   = {On Testing the Ratio of Scale Parameters of Two {P}areto Distributions Using a Novel Accelerated Sequential Sampling Technique with Applications},
  journal = {Sequential Analysis},
  volume  = {44},
  number  = {4},
  pages   = {464--490},
  year    = {2025}
}

\newpage
\renewcommand{\arraystretch}{0.45}

\begin{table}[H]
\centering
\caption{Set of Notations and Expressions $(i=1,2)$}
\label{tables_col_details}
\vspace{0.3cm}
\begin{tabular}{cl}\toprule
Notations & \multicolumn{1}{c}{Description}  \\ 
\midrule
$\rho_i$ &: Proportion, $\rho_i \in (0,1)$ defining the stopping boundary in the first stage. \\ \\
$k_i$ &: Number of observations per stage in the first phase. \\ \\
$n_i^*$ & : Optimal fixed sample size corresponding to the $i^{th}$ population.  \\ \\
$c_i$ & : Cost coefficient for $i^{th}$ population.\\ \\
$B$ & : Number of replications. \\ \\
$B'$ & : Number of bootstrap replications. \\ \\
$N_{ij}$ & : Sample size for population $i$ in the $j^{th}$ replication. \\ \\
$\displaystyle \bar{N}_i=\frac{1}{B}\sum_{j=1}^{B}n_{ij}$  & : Estimate of $n_i^*$. \\ \\ 
$\displaystyle s_{\bar{N}_i}=\sqrt{\frac{1}{B(B-1)}\sum_{j=1}^{B}(N_{ij}-\bar{N_i})^2}$  & : Standard error of $\bar{N}_i$. \\ \\
$\bar{N}_i/n_i^*$ & : Ratio of estimated to optimal sample size for population $i$. \\ \\
$\bar{N}_i - n_i^*$ & : Difference between estimated and optimal sample size for population $i$. \\ \\
$\hat{\alpha}_i$ & : $\hat{\alpha}_i= \begin{cases}
     1, & \text{if $H_0$ is rejected in $i^{th}$ run},\\
     0, & \text{if $H_0$ is accepted in $i^{th}$ run.}
 \end{cases}$ \\ \\
$\displaystyle\bar{\hat{\alpha}}= \frac{1}{B}\sum_{i=1}^{B}\hat{\alpha}_i$ & : Estimate of the type I error probability $\alpha$. \\ \\ 
$s_{\bar{\hat{\alpha}}}=\sqrt{\frac{\bar{\hat{\alpha}}(1-\bar{\hat{\alpha}})}{B}}$ & : SE of $\bar{\hat{\alpha}}$. \\ \\
$R_{\boldsymbol{n^*}}(\boldsymbol{\mu},\boldsymbol{c})$ & : Minimum fixed sample risk. \\ \\
$r_b$ & : Risk corresponding to the $b^{th}$ replication. \\ \\
$\bar{r}=\frac{\sum_{b=1}^{B}r_b}{B}$ & : Average risk efficiency. \\ \\
$\bar{\omega}=\frac{\bar{r}}{R_{\boldsymbol{n^*}}(\boldsymbol{\mu},\boldsymbol{c})}$ & : Average risk ratio. \\ \\
$s_{\bar{\omega}}$  & : SE of $\bar{\omega}$. \\ \\
$\bar{\kappa}=\bar{r}-R_{\boldsymbol{n^*}}(\boldsymbol{\mu},\boldsymbol{c}) $  & : Average regret.\\ \\
$s_{\bar{\kappa}}$  & : SE of $\bar{\kappa}$. \\ \\
$\xi_i$  & : Number of sampling operations for $i^{th}$ population. \\  \\
$\Psi^{*(j)}_{(N_1,N_2)}
=\left(X^{*(j)}_{N_1(1)}-Y^{*(j)}_{N_2(1)}
\right)$ & : Bootstrap replicate of
$\left(X_{N_1(1)} - Y_{N_2(1)}\right)$ from the $j^{{th}}$ bootstrap sample.\\ \\
$\displaystyle \overline{\Psi}^{*}_{(N_1,N_2)}=\frac{1}{B'}
\sum_{j=1}^{B'}\Psi^{*(j)}_{(N_1,N_2)}$ & : Bootstrap mean of the statistic $\left(X_{N_1(1)} - Y_{N_2(1)}\right)$.\\ \\
$\displaystyle s_{\Psi^{*}_{(N_1,N_2)}}=\sqrt{
\frac{1}{B'-1} \sum_{j=1}^{B'}
\left(\Psi^{*(j)}_{(N_1,N_2)}-\overline{\Psi}^{*}_{(N_1,N_2)}
\right)^2}$ & : Bootstrap SE of
$\left(X_{N_1(1)} - Y_{N_2(1)}\right)$.\\
\bottomrule
\end{tabular}
\end{table}

\renewcommand{\arraystretch}{0.90}

\begin{table}[H]
\centering
\caption{Simulation results under $H_0$ for the testing problem \eqref{hypothesis} using double sequential strategy \eqref{Stage_1}-\eqref{Stage_2} with $A=10^3, \alpha=0.05,\ \Delta_0=5,\ \Delta_1=6$}
\label{Sim_Under_H0}
\begin{tabular}{ccccccccccc}
\toprule
$(k_1, k_2)$ & $(\rho_1, \rho_2)$ & $(n_1^*, n_2^*)$ & $(c_1+c_2)$ &
$\bar{N}_1\ (s_{\bar{N}_1})$ & $\bar{N}_1/n_1^*$ & $\bar{N}_1-n_1^*$ & $\bar{\hat{\alpha}}$ & $\xi_1$ \\
 &  &  &  & $\bar{N}_2\ (s_{\bar{N}_2})$ & $\bar{N}_2/n_2^*$ & $\bar{N}_2-n_2^*$ & $ (s_{\bar{\hat{\alpha}}})$ & $\xi_2$ \\
\midrule 
(1,1) &  & (45,60) & 33.68973 & 44.56 (0.0711) & 0.9903 & -0.4352 & 0.0566 & 32.56 \\ 
 &  &  &  & 59.58 (0.0808) & 0.9930 & -0.4214 & (0.0023) & 47.58 \\
 &  & (150,200) & 0.00029 & 149.56 (0.1249) & 0.9970 & -0.4431 & 0.0521 & 137.56 \\
 &  &  &  & 199.81 (0.1435) & 0.9990 & -0.1918 & (0.0022) & 187.81 \\
 &  & (240,320) & $1.31 \times 10^{-8}$  & 239.55 (0.1568) & 0.9981 & -0.4536 & 0.0516 & 227.55 \\
 &  &  &  & 319.71 (0.1814) & 0.9991 & -0.2891 & (0.0022) & 307.71 \\
 &  & (450,600) & $9.64 \times 10^{-19}$ & 450.02 (0.2127) & 1.0000 & 0.0200 & 0.0499 & 438.02 \\
 &  &  &  & 599.48 (0.2455) & 0.9991 & -0.5165 & (0.0022) & 587.48 \\
\hline
(4,4) & (0.3,0.4) & (45,60) & 33.68973 & 44.58 (0.0714) & 0.9906 & -0.4250 & 0.0562 & 30.61 \\ 
 &  &  &  & 59.64 (0.0813) & 0.9940 & -0.3607 & (0.0023) & 38.45 \\
 &  & (150,200) & 0.00029 & 149.76 (0.1237) & 0.9984 & -0.2438 & 0.0507 & 112.77 \\
 &  &  &  & 199.49 (0.1434) & 0.9974 & -0.5143 & (0.0022) & 136.27 \\
 &  & (240,320) & $1.31 \times 10^{-8}$  & 239.88 (0.1585) & 0.9995 & -0.1173 & 0.0524 & 182.65 \\
 &  &  &  & 320.12 (0.1780) & 1.0004 & 0.1230 & (0.0022) & 220.79 \\
 &  & (450,600) & $9.64 \times 10^{-19}$ & 449.76 (0.2135) & 0.9995 & -0.2445 & 0.0491 & 345.15 \\
 &  &  &  & 599.42 (0.2468) & 0.9990 & -0.5789 & (0.0022) & 416.14 \\
\cline{2-9}
 & (0.5,0.7) & (45,60) & 33.68973 & 44.51 (0.0709) & 0.9890 & -0.4929 & 0.0592 & 24.44 \\ 
 &  &  &  & 59.52 (0.0807) & 0.9920 & -0.4804 & (0.0024) & 24.90 \\
 &  & (150,200) & 0.00029 & 149.69 (0.1235) & 0.9979 & -0.3143 & 0.0540 & 90.10 \\
 &  &  &  & 199.79 (0.1432) & 0.9989 & -0.2139 & (0.0023) & 91.54 \\
 &  & (240,320) & $1.31 \times 10^{-8}$  & 239.80 (0.1570) & 0.9992 & -0.1981 & 0.0518 & 146.45 \\
 &  &  &  & 320.06 (0.1794) & 1.0002 & 0.0609 & (0.0022) & 148.67 \\
 &  & (450,600) & $9.64 \times 10^{-19}$ & 449.71 (0.2126) & 0.9994 & -0.2853 & 0.0473 & 277.70 \\
 &  &  &  & 599.57 (0.2471) & 0.9993 & -0.4329 & (0.0021) & 281.36 \\
\hline
(6,6) & (0.3,0.4) & (45,60) & 33.68973 & 44.62 (0.0708) & 0.9916 & -0.3787 & 0.0542 & 29.78 \\ 
 &  &  &  & 59.76 (0.0803) & 0.9960 & -0.2413 & (0.0023) & 36.61 \\
 &  & (150,200) & 0.00029 & 149.69 (0.1235) & 0.9980 & -0.3052 & 0.0487 & 109.06 \\
 &  &  &  & 199.48 (0.1430) & 0.9974 & -0.5155 & (0.0022) & 129.68 \\
 &  & (240,320) & $1.31 \times 10^{-8}$  & 239.70 (0.1558) & 0.9988 & -0.2964 & 0.0525 & 176.56 \\
 &  &  &  & 319.69 (0.1777) & 0.9990 & -0.3084 & (0.0022) & 209.81 \\
 &  & (450,600) & $9.64 \times 10^{-19}$ & 449.50 (0.2138) & 0.9989 & -0.4979 & 0.0532 & 333.93 \\
 &  &  &  & 600.11 (0.2466) & 1.0002 & 0.1142 & (0.0022) & 396.75 \\
\cline{2-9}
 & (0.5,0.7) & (45,60) & 33.68973 & 44.51 (0.0718) & 0.9892 & -0.4879 & 0.0581 & 22.66 \\ 
 &  &  &  & 59.68 (0.0805) & 0.9947 & -0.3152 & (0.0023) & 21.51 \\
 &  & (150,200) & 0.00029 & 149.72 (0.1249) & 0.9981 & -0.2789 & 0.0529 & 84.13 \\
 &  &  &  & 199.62 (0.1429) & 0.9981 & -0.3838 & (0.0022) & 79.75 \\
 &  & (240,320) & $1.31 \times 10^{-8}$  & 239.66 (0.1570) & 0.9986 & -0.3373 & 0.0483 & 136.51 \\
 &  &  &  & 319.59 (0.1794) & 0.9987 & -0.4134 & (0.0021) & 129.86 \\
 &  & (450,600) & $9.64 \times 10^{-19}$ & 449.74 (0.2109) & 0.9994 & -0.2557 & 0.0526 & 259.19 \\
 &  &  &  & 599.46 (0.2453) & 0.9991 & -0.5395 & (0.0022) & 246.34 \\
\bottomrule
\end{tabular}
\end{table}

\begin{table}[H]
\centering
\setlength{\tabcolsep}{4pt}
\caption{Simulation results under $H_1$ for the testing problem \eqref{hypothesis} using double sequential strategy \eqref{Stage_1}-\eqref{Stage_2} with $A=10^3, \alpha=0.05,\ \Delta_0=4,\ \Delta_1=5$}
\label{Sim_Under_H1}
\begin{tabular}{ccccccccccc}\toprule
$(k_1, k_2)$ & $(\rho_1, \rho_2)$ & $(n_1^*, n_2^*)$ & $(c_1+c_2)$ &
$\bar{N}_1\ (s_{\bar{N}_1})$ & $\bar{N}_1/n_1^*$ & $\bar{N}_1-n_1^*$ & $\bar{\omega}$ & $\bar{\kappa}$ & $\xi_1$ \\
 &  &  &  & $\bar{N}_2\ (s_{\bar{N}_2})$ & $\bar{N}_2/n_2^*$ & $\bar{N}_2-n_2^*$ & $ (s_{\bar{\omega}})$ & $(s_{\bar{\kappa}})$ & $\xi_2$ \\
\midrule
(1,1) &  & (45,60) & 33.6897 & 44.72 (0.0712) & 0.9937 & -0.2848 & 1.0073 & 1.4676 & 32.72 \\
 &  &  &  & 59.71 (0.0801) & 0.9951 & -0.2918 & (0.0092) & (1.8648) & 47.71 \\
 &  & (120,160) & 0.0081 & 119.90 (0.1117) & 0.9991 & -0.1036 & 0.9290 & -0.0082 & 107.90 \\
 &  &  &  & 159.72 (0.1270) & 0.9983 & -0.2785 & (0.0006) & $(6.63 \times 10^{-5})$ & 147.72 \\
 &  & (270,360) & $4.68 \times 10^{-10}$  & 269.71 (0.1642) & 0.9989 & -0.2881 & 0.9669 & $-4.80 \times 10^{-10}$ & 257.71 \\
 &  &  &  & 359.78 (0.1921) & 0.9994 & -0.2233 & (0.0004) & $(5.67 \times 10^{-12})$ & 347.78 \\
 &  & (420,560) & $2.70 \times 10^{-17}$ & 419.59 (0.2056) & 0.9990 & -0.4135 & 0.9784 & $-2.78 \times 10^{-17}$ & 407.59 \\
 &  &  &  & 559.88 (0.2361) & 0.9998 & -0.1236 & $(0.0003)$ & $(4.08 \times 10^{-19})$ & 547.88 \\
\hline
(4,4) & (0.3,0.4) & (45,60) & 33.6897 & 44.60 (0.0710) & 0.9911 & -0.4022 & 1.0176 & 3.5615 & 30.59 \\
 &  &  &  & 59.65 (0.0814) & 0.9941 & -0.3533 & (0.0095) & (1.9228) & 38.42 \\
 &  & (120,160) & 0.0081 & 119.74 (0.1119) & 0.9978 & -0.2608 & 0.9281 & -0.0083 & 89.52 \\
 &  &  &  & 159.63 (0.1299) & 0.9977 & -0.3740 & (0.0006) & $(6.69 \times 10^{-5})$ & 108.42 \\
 &  & (270,360) & $4.68 \times 10^{-10}$  & 269.65 (0.1655) & 0.9987 & -0.3526 & 0.9668 & $-4.81 \times 10^{-10}$ & 205.67 \\
 &  &  &  & 359.80 (0.1883) & 0.9995 & -0.1955 & (0.0004) & $(5.64 \times 10^{-12})$ & 248.46 \\
 &  & (420,560) & $2.70 \times 10^{-17}$ & 419.33 (0.2053) & 0.9984 & -0.6669 & 0.9785 & $-2.77 \times 10^{-17}$ & 321.65 \\
 &  &  &  & 560.29 (0.2364) & 1.0005 & 0.2896 & $(0.0003)$ & $(4.06 \times 10^{-19})$ & 388.76 \\
\cline{2-10}
 & (0.5,0.7) & (45,60) & 33.6897 & 44.57 (0.0707) & 0.9905 & -0.4284 & 0.9977 & -0.4674 & 24.49 \\
 &  &  &  & 59.59 (0.0807) & 0.9932 & -0.4102 & (0.0090) & (1.8285) & 24.86 \\
 &  & (120,160) & 0.0081 & 119.79 (0.1111) & 0.9982 & -0.2113 & 1.7904 & 0.0917 & 71.35 \\
 &  &  &  & 159.81 (0.1276) & 0.9988 & -0.1897 & (0.8615) & (0.0999) & 72.44 \\
 &  & (270,360) & $4.68 \times 10^{-10}$  & 269.54 (0.1651) & 0.9983 & -0.4620 & 0.9669 & $-4.80 \times 10^{-10}$ & 165.11 \\
 &  &  &  & 359.97 (0.1915) & 0.9999 & -0.0269 & (0.0004) & $(5.64 \times 10^{-12})$ & 167.48 \\
 &  & (420,560) & $2.70 \times 10^{-17}$ & 419.87 (0.2070) & 0.9997 & -0.1349 & 0.9789 & $-2.72 \times 10^{-17}$ & 259.08 \\
 &  &  &  & 560.07 (0.2400) & 1.0001 & 0.0741 & $(0.0003)$ & $(4.11 \times 10^{-19})$ & 262.39 \\
\hline
(6,6) & (0.3,0.4) & (45,60) & 33.6897 & 44.61 (0.0708) & 0.9912 & -0.3944 & 1.0383 & 7.7336 & 29.79 \\
 &  &  &  & 59.62 (0.0796) & 0.9936 & -0.3836 & (0.0100) & (2.0210) & 36.49 \\
 &  & (120,160) & 0.0081 & 119.56 (0.1110) & 0.9963 & -0.4420 & 0.9278 & -0.0084 & 86.42 \\
 &  &  &  & 159.75 (0.1286) & 0.9984 & -0.2490 & (0.0006) & $(6.67 \times 10^{-5})$ & 103.43 \\
 &  & (270,360) & $4.68 \times 10^{-10}$  & 269.91 (0.1643) & 0.9997 & -0.0854 & 0.9673 & $-4.75 \times 10^{-10}$ & 199.10 \\
 &  &  &  & 359.76 (0.1926) & 0.9993 & -0.2393 & (0.0004) & $(5.63 \times 10^{-12})$ & 236.61 \\
 &  & (420,560) & $2.70 \times 10^{-17}$ & 419.90 (0.2054) & 0.9998 & -0.0970 & 0.9790 & $-2.71 \times 10^{-17}$ & 311.68 \\
 &  &  &  & 560.06 (0.2387) & 1.0001 & 0.0563 & $(0.0003)$ & $(4.07 \times 10^{-19})$ & 370.00 \\
\cline{2-10}
 & (0.5,0.7) & (45,60) & 33.6897 & 44.51 (0.0715) & 0.9890 & -0.4930 & 1.0159 & 3.2146 & 22.68 \\
 &  &  &  & 59.66 (0.0802) & 0.9944 & -0.3369 & (0.0095) & (1.9189) & 21.53 \\
 &  & (120,160) & 0.0081 & 119.56 (0.1110) & 0.9963 & -0.4418 & 0.9276 & -0.0084 & 66.48 \\
 &  &  &  & 159.68 (0.1293) & 0.9980 & -0.3210 & (0.0006) & $(6.66 \times 10^{-5})$ & 63.17 \\
 &  & (270,360) & $4.68 \times 10^{-10}$  & 269.83 (0.1660) & 0.9994 & -0.1675 & 0.9674 & $-4.73 \times 10^{-10}$ & 154.12 \\
 &  &  &  & 359.94 (0.1910) & 0.9998 & -0.0600 & (0.0004) & $(5.74 \times 10^{-12})$ & 146.64 \\
 &  & (420,560) & $2.70 \times 10^{-17}$ & 419.38 (0.2088) & 0.9985 & -0.6192 & 0.9779 & $-2.84 \times 10^{-17}$ & 241.18 \\
 &  &  &  & 559.60 (0.2392) & 0.9993 & -0.4022 & $(0.0003)$ & $(4.14 \times 10^{-19})$ & 229.67 \\
\bottomrule
\end{tabular}
\end{table}

\renewcommand{\arraystretch}{1}
\begin{table}[H]
    \centering
    \caption{Analysis of the the inter-arrival time data using double sequential procedure (\ref{Stage_1})-(\ref{Stage_2})}
    \label{real_data_analysis}
    \begin{tabular}{ccccccccc}
    \toprule
    $\Delta_1$ & $c_1+c_2$ & $(n_1^*, n_2^*)$ & $b$ & $N_1$ & $N_1/n_1^*$ & $X_{N_1(1)} - Y_{N_2(1)}$ & $\overline{\Psi}^{*}_{(N_1,N_2)}$ & Decision   \\
    &  &  &  & $N_2$ & $N_2/n_2^*$ & & $s_{\Psi^{*}_{(N_1,N_2)}}$ &   \\
    \midrule
    4 & 0.01 & (46.22, 55.67) & 0.930 & 50 & 1.081 & 0 & 0.138 & Fail to Reject $H_0$ \\
     &  &  &  & 58 & 1.041 &  & 0.423 &  \\
    \cline{2-9}
     & 0.001 & (56.97, 68.62) & 0.755 & 59 & 1.036 & 0 & 0.049 & Fail to Reject $H_0$ \\
     &  &  &  & 76 & 1.108 &  & 0.249 & \\
    \hline
    6 & 0.01 & (32.08, 38.63) & 1.340 & 35 & 1.091 & 0 & 0.276 & Fail to Reject $H_0$ \\
     &  &  &  & 41 & 1.061 &  & 1.564 & \\
    \cline{2-9}
     & 0.001 & (39.24, 47.26) & 1.095 & 40 & 1.019 & 0 & 0.670 & Fail to Reject $H_0$ \\
     &  &  &  & 46 & 0.973 &  & 1.102 & \\
    \bottomrule 
    \end{tabular}
\end{table}

\begin{figure}[H]
    \centering
    \includegraphics[width=0.5\linewidth]{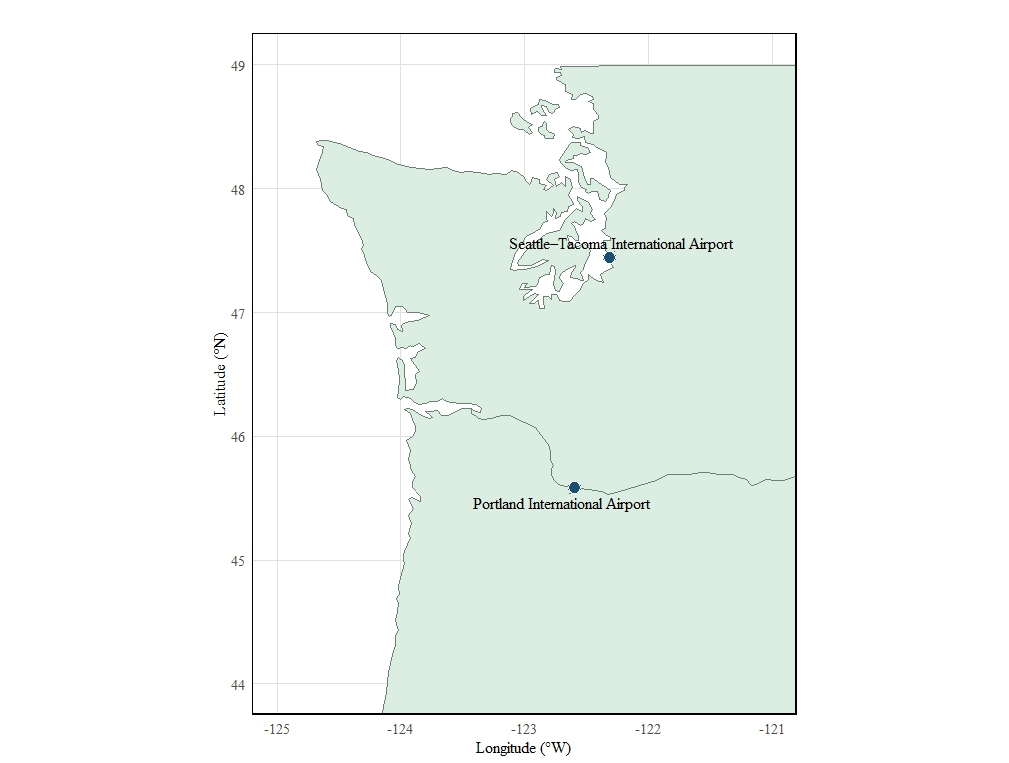}
    \caption{Geographical locations of meteorological stations}
    \label{location_map}
\end{figure}

\begin{figure}[H]
    \centering
    \subfigure[$\Delta_1=6$]{
    \includegraphics[width=0.4\linewidth]{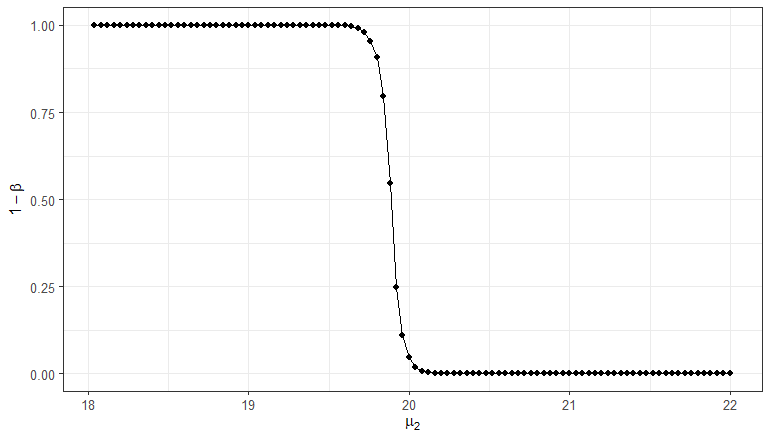}
    \label{power_curve_6}
    }
    \subfigure[$\Delta_1=7$]{
    \includegraphics[width=0.4\linewidth]{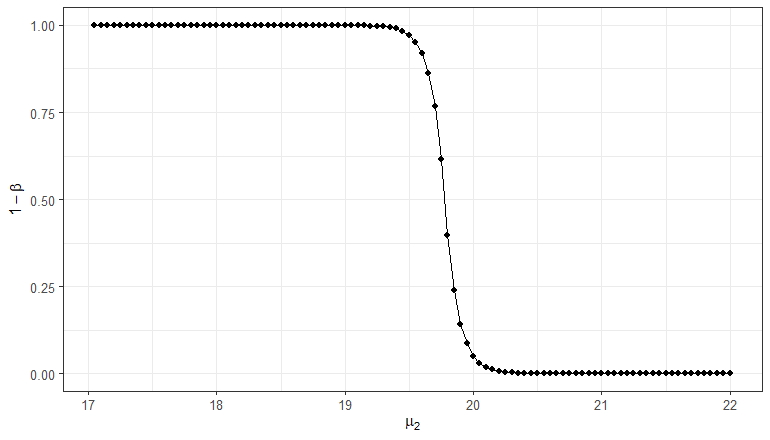}
    \label{power_curve_7}
    }    
    \caption{Power curve with respect to different values of $\mu_2$}
\end{figure}

\begin{figure}[H]
    \centering
    \subfigure[$(N_1,\,N_2)$ versus $(c_1+c_2)$]{
    \includegraphics[width=0.35\linewidth]{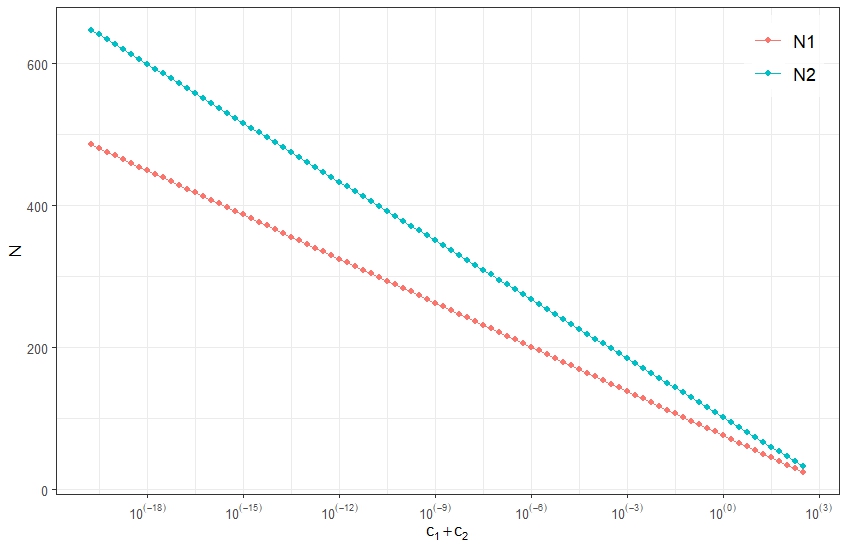}
    }
    \subfigure[$\alpha$ versus $(c_1+c_2)$]{
    \includegraphics[width=0.35\linewidth]{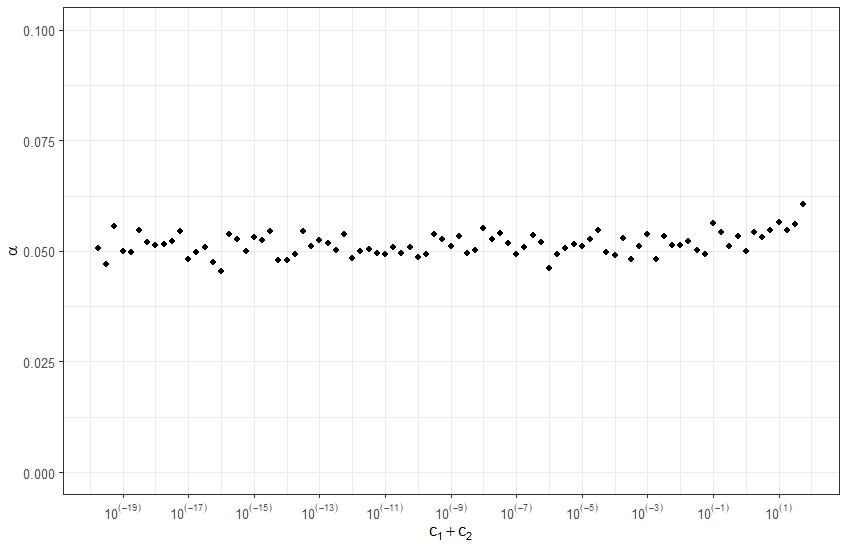}
    }\\
    \subfigure[$(1-\beta)$ versus $(c_1+c_2)$]{
    \includegraphics[width=0.35\linewidth]{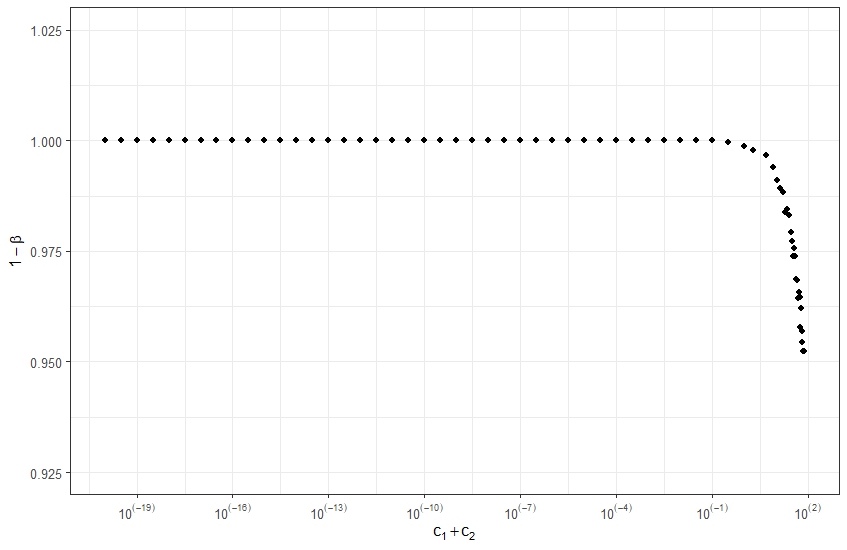}
    }
    \subfigure[$(\xi_1+\xi_2)$ versus $(c_1+c_2)$]{
    \includegraphics[width=0.35\linewidth]{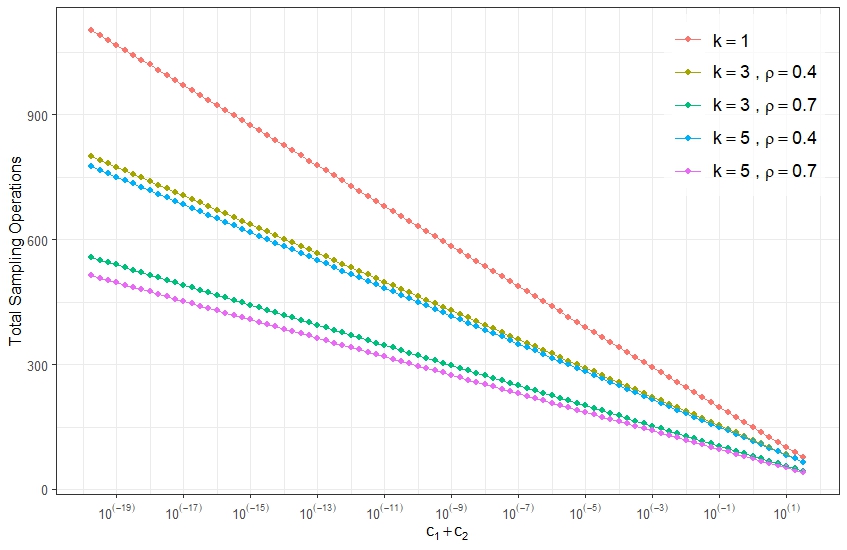}
    }
    \caption{Sensitivity of  $(N_1,\,N_2)$, $\alpha$, $(1-\beta)$ and $(\xi_1+\xi_2)$ for different values of $(c_1+c_2)$}
    \label{sensitivity_wrt_cost}
\end{figure}

\begin{figure}[H]
    \centering
    \subfigure[$(N_1,\,N_2)$ vs $m$]{
    \includegraphics[width=0.37\linewidth]{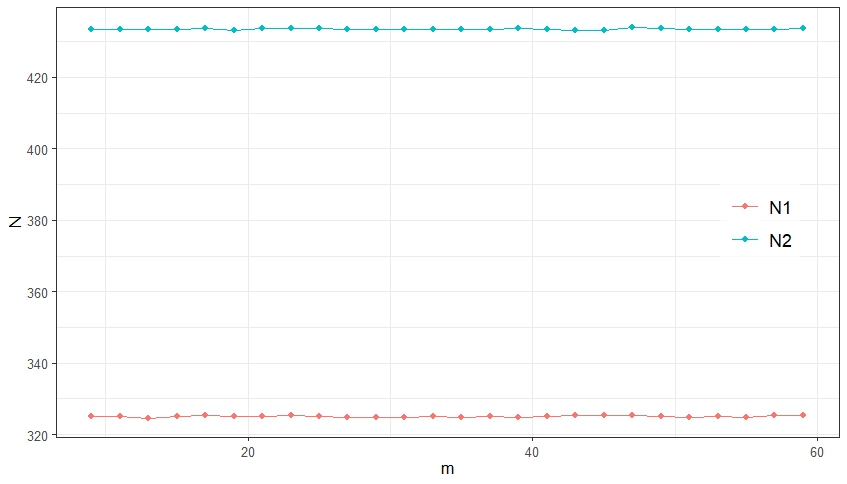}
    }
    \subfigure[$\alpha$ vs $m$]{
    \includegraphics[width=0.37\linewidth]{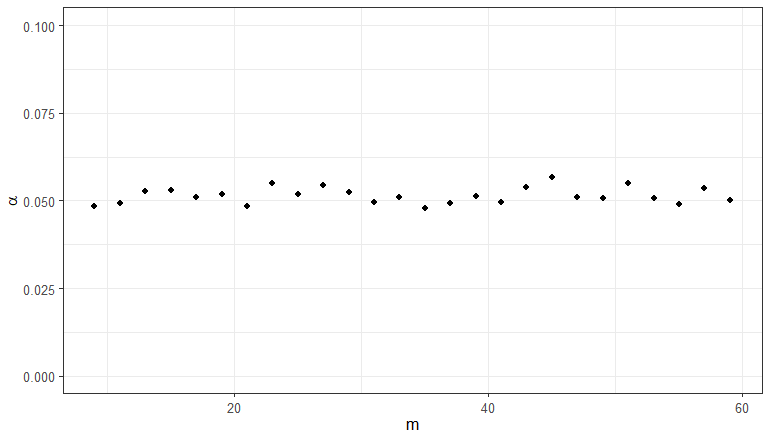}
    }\\
    \subfigure[Regret vs $m$]{
    \includegraphics[width=0.37\linewidth]{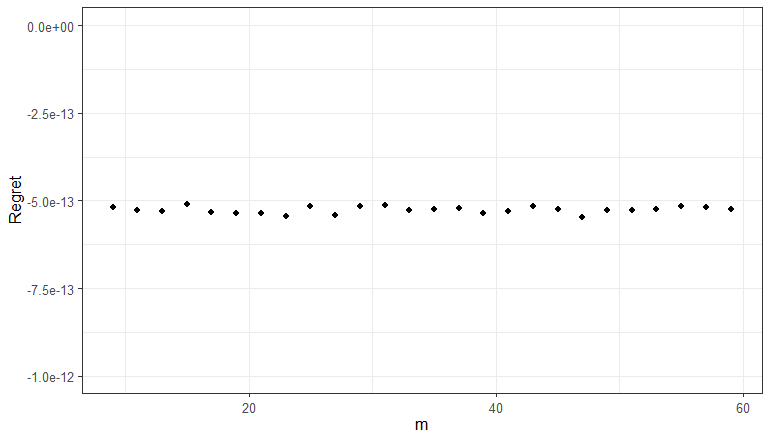}
    }
    \caption{Sensitivity of $(N_1,\,N_2)$, $\alpha$ and regret for different values of $m$}
    \label{sensitivity_wrt_m}
\end{figure}

\begin{figure}[H]
    \centering
    \subfigure[Seattle-Tacoma]{
    \includegraphics[width=0.40\textwidth]{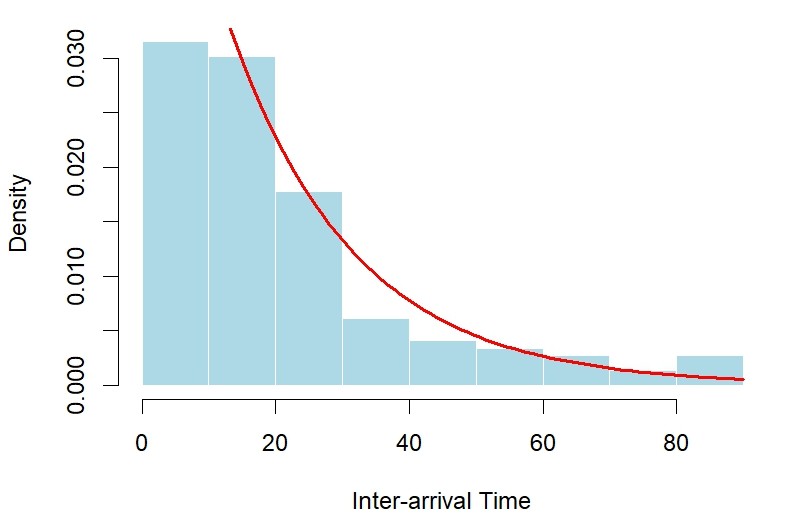}
    }
    \subfigure[Portland]{
    \includegraphics[width=0.40\textwidth]{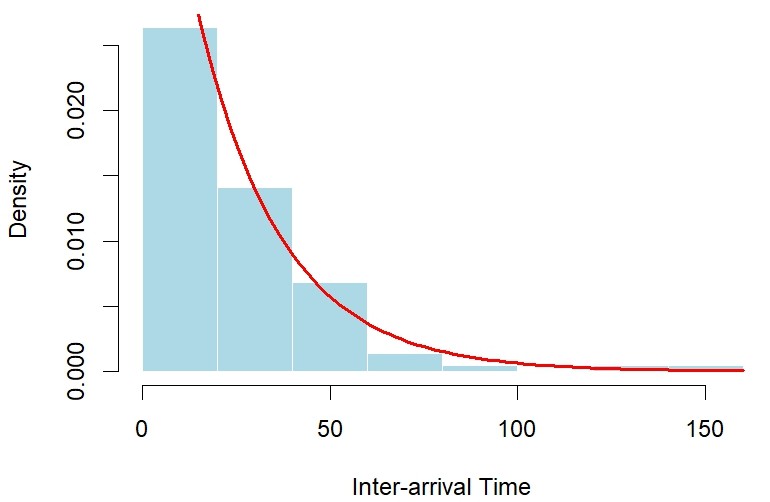}
    }
    \caption{Histograms for the inter-arrival times}
    \label{Histogram}
\end{figure}

\begin{figure}[H]
    \centering
    \subfigure[Seattle-Tacoma]{
    \includegraphics[width=0.40\textwidth]{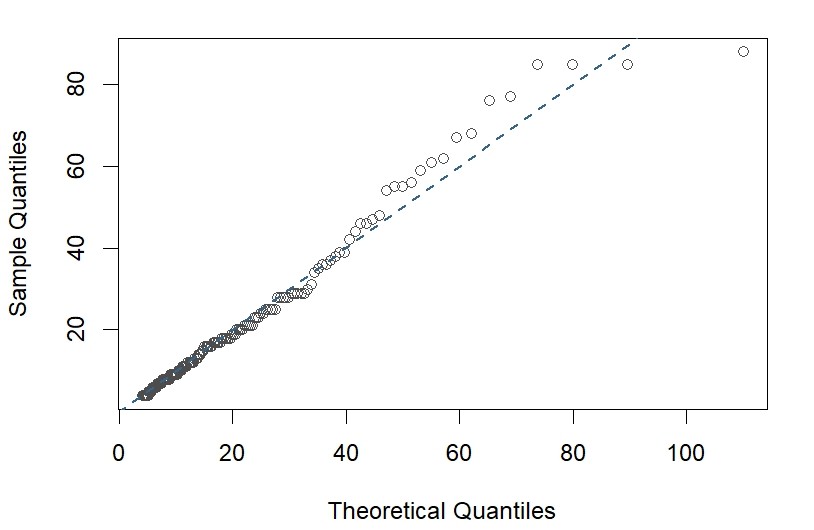}
    }
    \subfigure[Portland]{
    \includegraphics[width=0.40\textwidth]{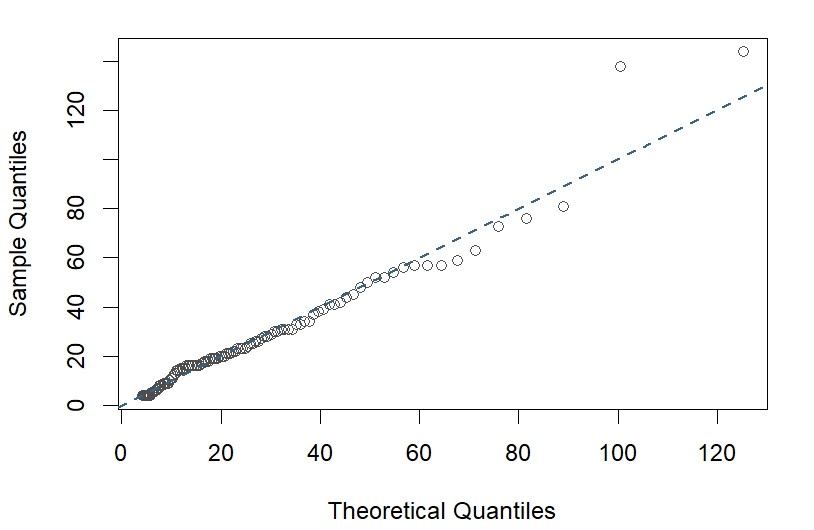}
    }
    \caption{Q-Q plots for the inter-arrival times}
    \label{QQPlot}
\end{figure}

\begin{figure}[H]
    \centering
    \subfigure[Seattle-Tacoma]{
    \includegraphics[width=0.40\textwidth]{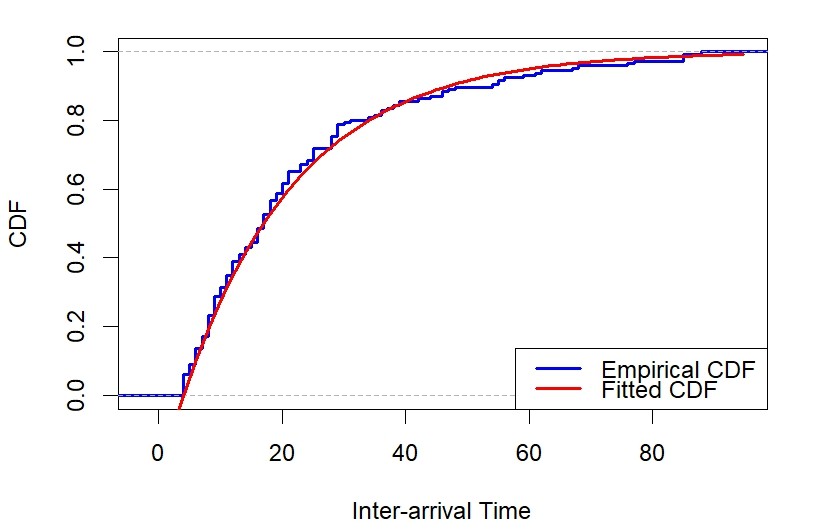}
    }
    \subfigure[Portland]{
    \includegraphics[width=0.40\textwidth]{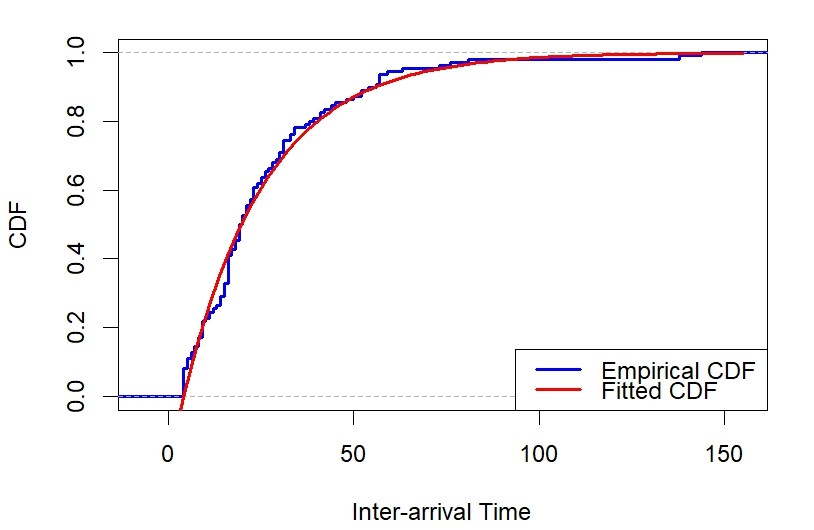}
    }
    \caption{ECDF plots of the inter-arrival times}
    \label{ECDF}
\end{figure}

\clearpage

\appendix
\section*{Appendix: Proofs}
\label{Appendix}
To establish the first- and second-order asymptotic properties of the proposed double-sequential sampling scheme, it is instructive to first consider the corresponding purely sequential sampling scheme defined in \eqref{purely_seq_rule} along with its associated results.
\bigskip\\
\indent Under assumption \eqref{pilot_sample_assumption}, for $k_1=k_2 = 1$, we can write \eqref{purely_seq_rule} as
\begin{align}
\label{woodroofe_eqn}
    N_i' &= \inf \left\{ n \geq m_{i0}+1 : n \geq U^{(i)}_{n} d_{\alpha} \right\},\notag\\
    N_i'^* &= \inf \left\{ n \geq m_{i0} : (n+1) \geq U^{(i)}_{n+1} d_{\alpha} \right\}, \,\ \text{where } N_i'^{*} = N_i' - 1 \notag\\
    &= \inf \left\{ n \geq m_{i0} : \sum_{i=1}^{n} W_i \leq \left( \frac{1}{n_i^*} \right) n^2 \left( 1 + \frac{1}{n} \right) \right\}.
\end{align}
where $W_1, W_2, \dots, W_n$ are independent and identically distributed (i.i.d.) $\mathrm{Exponential}(1)$ random variables. In particular, $\mathbb{P}(W_i \leq w) = 1 - e^{-w} \leq w, \,\ \forall\, w > 0.$
\bigskip\\
\indent By comparing \eqref{woodroofe_eqn} with Equation~(1.1) of \textcite{woodroofe_1977}, and using the notation therein, we obtain
\[
S_n = \sum_{i=1}^{n} W_i,\ c = \frac{1}{n_i^*}, \ \alpha = 2, \ \beta = 1, \ \mu = 1, \ \tau^2 =1,\lambda = n_i^*,\ L(n) = \left(1 + \frac{1}{n} \right), \ L(0) = 1, \ a=1, \ b=1.
\]
Since \( N_i' = N_i'^* + 1 \) almost surely, it follows from \textcite{woodroofe_1977} that, as $q \to \infty$ with $q = \log \left(\frac{1}{c_1 + c_2}\right)$,
\begin{align}
\label{prop1_purely}
    \begin{aligned}
        \text{ (i) }\quad & \mathbb{E}_{\boldsymbol{\theta}}(N_i') = n_i^* + a_2 + o(1)\ \text{ where }\ a_2 = \nu_\alpha(1,1) - 1 \approx -1/4 ,\\
        \text{ (ii) }\quad & Q_i'=\frac{N_i' - n_i^*}{\sqrt{n_i^*}} \xrightarrow{d} \mathcal{N}(0,1),\\
        \text{ (iii) }\quad & \mathbb{P}_{\boldsymbol{\theta}}(N_i'\leq \epsilon n_i^*) = O(n_i^{*-(m-1)})\ \text{ if $0<\epsilon<1$ fixed},\\
        \text{ (iv) }\quad & \left|Q_i'\right|^s\ \text{ is uniformly integrable $\forall s>0$ if }  m>s/2+1.
    \end{aligned}
\end{align}

\noindent Next, we know that, for all fixed $n_i\geq m,\ i=1,2$,\ $X_{n_1(1)}$ and $Y_{n_2(1)}$ are independent of $(U^{(1)}_{m},\ U^{(1)}_{m+1}, \ldots, U^{(1)}_{n_1})$ and $(U^{(2)}_{m},\ U^{(2)}_{m+1}, \ldots, U^{(2)}_{n_2})$. Let $\boldsymbol{N'}=(N_1', N_2')$ then, under the purely sequential procedure defined in \eqref{purely_seq_rule}, for all $n_i\geq m,\ i=1,2$, the event $I(\boldsymbol{N'}=\boldsymbol{n})$ depends solely on $(U^{(1)}_{m},\ U^{(1)}_{m+1}, \ldots, U^{(1)}_{n_1})$ and $(U^{(2)}_{m},\ U^{(2)}_{m+1}, \ldots, U^{(2)}_{n_2})$. Therefore, we can conclude that
\begin{equation}
\label{indep_of_N=n_and_min_o.s.}
    I(\boldsymbol{N'}=\boldsymbol{n})\text{ is independent of } \left(X_{n_1(1)},\,Y_{n_2(1)}\right).
\end{equation}
Now, the probability of type I error can be expressed as
\begin{equation}
\label{type-I error for random N}
    \begin{split}
        \mathbb{P}_{\boldsymbol{\theta,N'}}\left(\text{Type I error}\right) &= \mathbb{E}_{\boldsymbol{\theta}}\left[\mathbb{P}_{\boldsymbol{\theta}}\left(X_{N_1'(1)}-Y_{N_2'(1)} > b \mid H_0 \right)\right]\\
        &= \mathbb{E}_{\boldsymbol{\theta}}\left[\left(\frac{\sigma_1}{N_1'}+\frac{\sigma_2}{N_2'}\right)^{-1}\frac{\sigma_1}{N_1'}\exp\left\{-\frac{N_1'(b-\Delta_0)}{\sigma_1} \right\} \right],
    \end{split}
\end{equation}
Let $U=\frac{N_1'}{n_1^*}$, $V=\frac{N_2'}{n_2^*}$ and $l_1= d_\alpha(b-\Delta_0) >0$, then
\begin{equation}
    \mathbb{P}_{\boldsymbol{\theta,N'}}\left(\text{Type I error}\right) = \mathbb{E}_{\boldsymbol{\theta}}\left[\left(1+\frac{U}{V}\right)^{-1}\exp(-l_1U) \right]=\mathbb{E}_{\boldsymbol{\theta}}\left[g(U,V)\right].
\end{equation}
Next, applying Taylor’s series expansion around $(U,V) = (1,1)$, we get
\begin{equation}
\label{type-I_error_2nd_order_exp1}
    \begin{split}
        \mathbb{P}_{\boldsymbol{\theta,N'}}\left(\text{Type I error}\right) &= \mathbb{E}_{\boldsymbol{\theta}}\Bigg[g(1,1) + (U-1)\left.\frac{\partial g}{\partial U}\right|_{(1,1)}+ (V-1)\left.\frac{\partial g}{\partial V}\right|_{(1,1)}+ \frac{1}{2} (U-1)^2 \left.\frac{\partial^2 g}{\partial U^2}\right|_{(1,1)}\\
        &\qquad + (U-1)(V-1) \left.\frac{\partial^2 g}{\partial U \partial V}\right|_{(1,1)}+ \frac{1}{2}(V-1)^2 \left.\frac{\partial^2 g}{\partial V^2}\right|_{(1,1)}\Bigg] + R_1(q^{-1}) \\
        &= \alpha + \frac{\alpha}{d_\alpha}\left[-\frac{l_1a_2}{\sigma_1}-\frac{a_2}{2\sigma_1}+\frac{a_2}{2\sigma_2}+\frac{l_1^2}{2\sigma_1}+\frac{l_1}{2\sigma_1}+\frac{1}{4\sigma_1}-\frac{1}{4\sigma_2}-\frac{l_1a_2^2}{2\sigma_1\sigma_2d_\alpha}\right]+ R_1(q^{-1}),
\end{split}
\end{equation}
where, 
\begin{align*}
R_1(q^{-1}) 
&= \mathbb{E}_{\boldsymbol{\theta}}\Big[R_{11} + R_{12} + R_{13} + R_{14}\Big]\\
&= \mathbb{E}_{\boldsymbol{\theta}}\Bigg[\frac{1}{6}
(U-1)^3 \left.\frac{\partial^3 g}{\partial U^3}\right|_{(S_1,S_2)} + \frac{1}{2}(U-1)^2(V-1) \left.\frac{\partial^3 g}{\partial U^2 \partial V}\right|_{(S_1,S_2)} + \frac{1}{2}(U-1)(V-1)^2 \left.\frac{\partial^3 g}{\partial U \partial V^2}\right|_{(S_1,S_2)}\\ 
&\hspace{1 cm} + \frac{1}{6}(V-1)^3 \left.\frac{\partial^3 g}{\partial V^3}\right|_{(S_1,S_2)}
\Bigg],
\end{align*}
with $S_1$ and $S_2$ being intermediate random variables lying between $1$ and $U$, and $1$ and $V$, respectively.
\bigskip\\
\indent Let $H_1 = (S_1^{-1} + S_2^{-1}),$ then
\begin{align*}
    R_{11} &=  \sigma_1^{-3/2} d_\alpha^{-3/2} \exp(-l_1S_1) \Big[ -S_1^{-4}S_2^{-3}H_1^{-4}-S_1^{-3}S_2^{-2}l_1H_1^{-3}-\frac{1}{2}S_1^{-2}S_2^{-1}l_1^2H_1^{-2}-\frac{1}{6}S_1^{-1}S_2^{0}l_1^3H_1^{-1}\Big]\, Q_1'^3,\\    
    R_{12} &= \sigma_1^{-1}\sigma_2^{-1/2}  d_\alpha^{-3/2} \exp(-l_1S_1) \Big [-2S_1^{-3}S_2^{-3}H_1^{-3}+3S_1^{-3}S_2^{-4}H_1^{-4}-S_1^{-2}S_2^{-2}l_1H_1^{-2}+2S_1^{-2}S_2^{-3}l_1H_1^{-3}\\
    &\hspace{5cm}+\frac{1}{2}S_1^{-1}S_2^{-2}l_1^2H_1^{-2}\Big]\, Q_1'^2Q_2', \\
    R_{13} &= \sigma_1^{-1/2}\sigma_2^{-1}  d_\alpha^{-3/2} \exp(-l_1S_1)\Big [-S_1^{-3}S_2^{-3}H_1^{-3}+3S_1^{-3}S_2^{-4}H_1^{-4}+S_1^{-2}S_2^{-3}l_1H_1^{-3}\Big]\, Q_2'^2Q_1', \\
    R_{14} &= \sigma_2^{-3/2} d_\alpha^{-3/2} \exp(-l_1S_1)\Big[S_1^{-3}S_2^{-4}H_1^{-4}\Big]\, Q_2'^3.
\end{align*}
We use the fact that, for $a>0$ and two random variables $X, Y > 0$ almost surely, one has
\begin{equation}
    \label{H_bound}
    (X^{-1} + Y^{-1})^{-j} e^{-a Y} \le \frac{i!}{a^iY^{(i-j)}} \,\ \text{almost surely},
\end{equation}
where $i, j \in \mathbb{N}$ are arbitrary natural numbers with $i \ge j$. Thus
\begin{align*}
    |R_{11}| &\leq  \sigma_1^{-3/2} d_\alpha^{-3/2} \Big[ 24S_1^{-4}S_2^{-3}l_1^{-4}+6S_1^{-3}S_2^{-2}l_1^{-2}+S_1^{-2}S_2^{-1}+\frac{1}{6}S_1^{-1}S_2^{0}l_1^{2}\Big]\, \left|Q_1'^3\right|,\\    
    |R_{12}| &\leq \sigma_1^{-1}\sigma_2^{-1/2}  d_\alpha^{-3/2}\Big [12S_1^{-3}S_2^{-3}l_1^{-3}+72S_1^{-3}S_2^{-4}l_1^{-4}+2S_1^{-2}S_2^{-2}l_1^{-1}+12S_1^{-2}S_2^{-3}l_1^{-2}+S_1^{-1}S_2^{-2}\Big]\, \left|Q_1'^2Q_2'\right|, \\
    |R_{13}| &\leq \sigma_1^{-1/2}\sigma_2^{-1}  d_\alpha^{-3/2} \Big [6S_1^{-3}S_2^{-3}l_1^{-3}+72S_1^{-3}S_2^{-4}l_1^{-4}+6S_1^{-2}S_2^{-3}l_1^{-2}\Big]\, \left|Q_2'^2Q_1'\right|, \\
    |R_{14}| &\leq \sigma_2^{-3/2} d_\alpha^{-3/2} \Big[24S_1^{-3}S_2^{-4}l_1^{-4}\Big]\, \left|Q_2'^3\right|.
\end{align*}
Now, we decompose the sample space into events $E_1', E_2', E_3', E_4'$ based on whether $N_i > \varepsilon C_i$ or $N_i \le \varepsilon C_i$, where $0 < \epsilon < 1$ and 
\begin{equation}
    \label{event_decomposition}
    \begin{split}
        E_1' = [N_1' > \epsilon n^*_1] \cap [N_2' > \epsilon n^*_2], \quad E_2' = [N_1' > \epsilon n^*_1] \cap [N_2' \le \epsilon n^*_2],\\
    E_3' = [N_1' \le \epsilon n^*_1] \cap [N_2' > \epsilon n^*_2], \quad E_4' = [N_1' \le \epsilon n^*_1] \cap [N_2' \le \epsilon n^*_2].
    \end{split}
\end{equation}
\textbf{Case I: ($E_1'$).} Here, $S_1$ and $S_2$ are uniformly bounded away from zero, ensuring that their inverse powers remain bounded. Moreover, $\mathbb{E}_{\boldsymbol{\theta}}\left(|Q_1'|^3\right) = O(1)$ (by \ref{prop1_purely}) provided that $m > 5/2$. Thus, for any fixed positive constant $B_{11}$ independent of $q$, we have
\begin{equation}
    \label{R1_under_good_events1}
    \big|\mathbb{E}_{\boldsymbol{\theta}}[R_{11}I(E_1')]\big| \leq B_{11}d_\alpha^{-3/2}\mathbb{E}_{\boldsymbol{\theta}}\left[|Q_1'|^3\right]=O(q^{-3/2})=o(q^{-1}).
\end{equation}
Similarly
\begin{equation}
    \label{R1_under_good_events2}
    \big|\mathbb{E}_{\boldsymbol{\theta}}[R_{1i}I(E_1')] \big|= o(q^{-1}),\quad i=2,3,4.
\end{equation} 
\textbf{Case II: ($E_2', E_3', E_4'$).} 
In this case, $S_1$ and $S_2$ may approach zero; however, the probabilities of the corresponding events satisfy
\[
\mathbb{P}_{\boldsymbol{\theta}}(E_2')=\mathbb{P}_{\boldsymbol{\theta}}(E_3')= O\!\left(q^{-(m-1)}\right) \,\ \text{and} \,\ 
\mathbb{P}_{\boldsymbol{\theta}}(E_4')=O\!\left(q^{-2(m-1)}\right),
\]
which are asymptotically negligible.
\bigskip\\
\indent Therefore, under the condition $m \geq 8$, it follows that
\begin{equation}
\label{R1_under_bad_events}
    \big|\mathbb{E}_{\boldsymbol{\theta}}[R_{1i}I(E_j)]\big|= o(q^{-1}),\quad i=1,2,3,4,\ j=2,3,4. 
\end{equation}
Combining \eqref{R1_under_good_events1}-\eqref{R1_under_bad_events}, we conclude that
\begin{equation}
    \label{|R_1|_bound}
    \left|R_1(q^{-1})\right| \leq \sum_{i=1}^{4}\left|\mathbb{E}_{\boldsymbol{\theta}}(R_{1i})\right| = o(q^{-1}).
\end{equation}
Finally, using \eqref{type-I_error_2nd_order_exp1} and $a_2 \approx -\tfrac{1}{4}$, we obtain
\begin{equation}
\label{type-I_error_2nd_order_exp2}
        \mathbb{P}_{\boldsymbol{\theta,N'}}\left(\text{Type I error}\right) \approx \alpha + \frac{\alpha}{d_\alpha}\left[\frac{3}{8\sigma_1}-\frac{3}{8\sigma_2}+\frac{l_1^2}{2\sigma_1}+\frac{3l_1}{4\sigma_1}-\frac{l_1}{32\sigma_1\sigma_2d_\alpha}\right] + o(q^{-1}).
\end{equation}

\noindent Now, in a manner similar to that in \eqref{type-I error for random N}–\eqref{type-I_error_2nd_order_exp2}, the risk function is given by
\begin{equation}
    R_{\boldsymbol{N'}}(\boldsymbol{\mu},\boldsymbol{c}) = \mathbb{E}_{\boldsymbol{\theta}}\left[ A \left(\frac{\sigma_1}{N_1'}+\frac{\sigma_2}{N_2'}\right)^{-1}\frac{\sigma_2}{N_2'}\exp\left\{-\frac{N_2'(\Delta_1-b)}{\sigma_2} \right\} + c_1 N_1' \sigma_1^{-1} + c_2 N_2' \sigma_2^{-1}\right].
\end{equation}
Let $l_2= d_\alpha(\Delta_1-b)=\log \left\{\frac{A(\Delta_1-\Delta_0)}{2(c_1+c_2)}\right\} >0$, then
\begin{align*}
   R_{\boldsymbol{N'}}(\boldsymbol{\mu},\boldsymbol{c}) &= \mathbb{E}_{\boldsymbol{\theta}}\left[ A \left(1+\frac{V}{U}\right)^{-1}\exp(-l_2V) + \left(c_1 U + c_2 V\right)d_\alpha\right]=\mathbb{E}_{\boldsymbol{\theta}}\left[h(U,V)\right].
\end{align*}
Next, applying Taylor’s series expansion around $(U,V) = (1,1)$, we get
\begin{equation}
    \label{regret_2nd_order_exp1}
    \begin{split}
        R_{\boldsymbol{N'}}(\boldsymbol{\mu},\boldsymbol{c}) &= R_{\boldsymbol{n^*}}(\boldsymbol{\mu},\boldsymbol{c})+\frac{A\exp(-l_2)}{4 d_\alpha}\left[\frac{a_2}{\sigma_1}-\frac{2l_2a_2}{\sigma_2}-\frac{a_2}{\sigma_2}-\frac{1}{2\sigma_1}+\frac{l_2^2}{\sigma_2}+\frac{l_2}{\sigma_2}+\frac{1}{2\sigma_2}-\frac{l_2a_2^2}{\sigma_1\sigma_2d_\alpha}\right]\\
        &\quad +\left(\frac{c_1}{\sigma_1}+\frac{c_2}{\sigma_2}\right)a_2 + R_2(q^{-1}),
    \end{split}
\end{equation}
where
\begin{align*}
R_2(q^{-1}) 
&= \mathbb{E}_{\boldsymbol{\theta}}\Big[R_{21} + R_{22} + R_{23} + R_{24}\Big]\\
&= \mathbb{E}_{\boldsymbol{\theta}}\Bigg[\frac{1}{6}
(U-1)^3 \left.\frac{\partial^3 h}{\partial U^3}\right|_{(M_1,M_2)} + \frac{1}{2}(U-1)^2(V-1) \left.\frac{\partial^3 h}{\partial U^2 \partial V}\right|_{(M_1,M_2)} + \frac{1}{2}(U-1)(V-1)^2 \left.\frac{\partial^3 h}{\partial U \partial V^2}\right|_{(M_1,M_2)}\\ 
&\hspace{1 cm} + \frac{1}{6}(V-1)^3 \left.\frac{\partial^3 h}{\partial V^3}\right|_{(M_1,M_2)}
\Bigg],
\end{align*}
where $M_1$ and $M_2$ are intermediate random variables lying between $1$ and $U$, and $1$ and $V$, respectively.
\bigskip\\
\indent Define $H_2 = M_1^{-1} + M_2^{-1}$, then
\begin{align*}
    R_{21} &=  A\sigma_1^{-3/2} d_\alpha^{-3/2} \exp(-l_2M_2) \Big[ M_1^{-4}M_2^{-3}H_2^{-4}\Big]Q_1'^3,\\
    R_{22} &= A\sigma_1^{-1}\sigma_2^{-1/2}  d_\alpha^{-3/2} \exp(-l_2M_2) \Big [- M_1^{-3}M_2^{-3}H_2^{-3} + 3 M_1^{-4}M_2^{-3}H_2^{-4} + M_1^{-3}M_2^{-2}l_2H_2^{-3}\Big] Q_1'^2Q_2', \\
    R_{23} &= A\sigma_1^{-1/2}\sigma_2^{-1}  d_\alpha^{-3/2} \exp(-l_2M_2)\Big [-2M_1^{-3}M_2^{-3}H_2^{-3} +3 M_1^{-4}M_2^{-3}H_2^{-4} -  M_1^{-2}M_2^{-2}l_2H_2^{-2}+ 2 M_1^{-3}M_2^{-2}l_2H_2^{-3} \\ 
    &\hspace{5.5cm} + \frac{1}{2} M_1^{-2}M_2^{-1}l_2^2H_2^{-2}\Big] Q_2'^2Q_1', \\
    R_{24} &= A\sigma_2^{-3/2} d_\alpha^{-3/2} \exp(-l_2M_2)\Big[- M_1^{-3}M_2^{-4}H_2^{-4} - M_1^{-2}M_2^{-3}l_2H_2^{-3} - \frac{1}{2} M_1^{-1}M_2^{-2}l_2^2H_2^{-2}-\frac{1}{6} M_1^{0}M_2^{-1}l_2^3 H_2^{-1}\Big] Q_2'^3.
\end{align*}
Using the fact \eqref{H_bound}, we have
\begin{align*}
    |R_{21}|&\leq A\sigma_1^{-3/2} d_\alpha^{-3/2} \Big[24M_1^{-4}M_2^{-3}l_2^{-4}\Big]\left|Q_1'^3\right|,\\
    |R_{22}| &\leq A\sigma_1^{-1}\sigma_2^{-1/2}   d_\alpha^{-3/2} \Big [ 6M_1^{-3}M_2^{-3}l_2^{-3} + 72 M_1^{-4}M_2^{-3}l_2^{-4} + 6M_1^{-3}M_2^{-2}l_2^{-2}\Big]\left|Q_1'^2Q_2'\right|,\\
    |R_{23}|&\leq A\sigma_1^{-1/2}\sigma_2^{-1}  d_\alpha^{-3/2}\Big [12M_1^{-3}M_2^{-3}l_2^{-3} +72 M_1^{-4}M_2^{-3}l_2^{-4} + 2 M_1^{-2}M_2^{-2}l_2^{-1}+ 12 M_1^{-3}M_2^{-2}l_2^{-2} + M_1^{-2}M_2^{-1}l_2^{0}\Big]\left|Q_2'^2Q_1'\right|,  \\
    |R_{24}| &\leq A\sigma_2^{-3/2} d_\alpha^{-3/2} \Big [24 M_1^{-3}M_2^{-4}l_2^{-4} + 6 M_1^{-2}M_2^{-3}l_2^{-2} + M_1^{-1}M_2^{-2}l_2^{0}+M_1^{0}M_2^{-3}l_2^{0}\Big]\left|Q_2'^3\right|.
\end{align*}
Considering the same decomposition of events \eqref{event_decomposition}, for any fixed positive constant $B_{21}$ independent of $q$, we have
\begin{equation}
    \label{R2_under_good_events1}
    \big|\mathbb{E}_{\boldsymbol{\theta}}[R_{21}I(E_1')]\big| \leq B_{21}d_\alpha^{-3/2}l_2^{-4}\mathbb{E}_{\boldsymbol{\theta}}\left[|Q_1'|^3\right]=o(q^{-1}).
\end{equation}
Similarly
\begin{equation}
    \label{R2_under_good_events2}
    \big|\mathbb{E}_{\boldsymbol{\theta}}[R_{2i}I(E_1')] \big|= o(q^{-1}),\,\ i=2,3,4,
\end{equation} 
and for $m \geq 4$, we have 
\begin{equation}
\label{R2_under_bad_events}
    \big|\mathbb{E}_{\boldsymbol{\theta}}[R_{2i}I(E_j)]\big|= o(q^{-1}),\quad i=1,2,3,4,\ j=2,3,4. 
\end{equation}
On combining \eqref{R2_under_good_events1}-\eqref{R2_under_bad_events}, we can conclude that 
\begin{equation}
    \left|\mathbb{E}_{\boldsymbol{\theta}}(R_2)\right| \leq \sum_{i=1}^{4}\left|\mathbb{E}_{\boldsymbol{\theta}}(R_{2i})\right| = o(q^{-1}).
\end{equation}
Thus, using \eqref{regret_2nd_order_exp1} and $a_2 \approx -\tfrac{1}{4}$, the regret can be expressed as
\begin{equation}
    \label{regret_2nd_order_exp2}
    \begin{split}
        \text{(Regret)}_{\boldsymbol{N'}}&=R_{\boldsymbol{N'}}(\boldsymbol{\mu},\boldsymbol{c})-R_{\boldsymbol{n^*}}(\boldsymbol{\mu},\boldsymbol{c})\\ 
        &\approx \frac{A\exp(-l_2)}{4 d_\alpha}\left[-\frac{3}{4\sigma_1}+\frac{3}{4\sigma_2}+\frac{l_2^2}{\sigma_2}+\frac{3l_2}{2\sigma_2}-\frac{l_2}{16\sigma_1\sigma_2d_\alpha}\right]-\frac{1}{4}\left(\frac{c_1}{\sigma_1}+\frac{c_2}{\sigma_2}\right) + o(q^{-1}).
    \end{split}
\end{equation}

\begin{proof}[\textbf{Proof of Lemma \ref{lemma1}}]
    For a given $0<\epsilon<\rho^{-1}-1$, consider
    \[
    \mathbb{P}_{\boldsymbol{\theta}}\{N_i' < T_i\} \leq \ \mathbb{P}_{\boldsymbol{\theta}}\{N_i' \leq (1+\epsilon)\rho n_i^*\} + \mathbb{P}_{\boldsymbol{\theta}}\{T_i > (1+\epsilon)\rho n_i^* \}.
    \]
    From \eqref{prop1_purely}, we have 
    \begin{equation}
    \label{Prob1_lemma1}
        \mathbb{P}_{\boldsymbol{\theta}}\{N_i' \leq (1+\epsilon)\rho n_i^*\} \leq O(n_i^{*-(m-1)}).
    \end{equation}
    Under assumption \eqref{pilot_sample_assumption}, rewrite the stopping rule \eqref{Stage_1} as 
    \[
    L_i^* = \inf \left\{ n \geq m_{i0} : \sum_{i=1}^{n} D_i \leq \left(\frac{2k^2}{\rho n_i^*}\right)n^2\left(1+\frac{1}{kn}\right) \right\},
    \]
    where, $L_i^* = L_i + m_{i0}$ and $D_1,D_2,\ldots,D_n\overset{\text{i.i.d.}}{\sim} \chi^2_{(2k)}$. As $T_i = kL_i^* + 1$ with probability 1, it follows from \textcite{woodroofe_1977} that, as $q \to \infty$
    \begin{align}
        \begin{aligned}
        \label{Ti_prop}
        \text{(i)}\quad & (\rho n_i^*)^{-1/2}(T_i-\rho n_i^*) \xrightarrow{d} \mathcal{N}(0,1),\\
        \text{(ii)}\quad & \big |(\rho n_i^*)^{-1/2}(T_i-\rho n_i^*) \big |^s \text{ is uniformly integrable } \forall s > 0 \text{ if } m>\frac{s}{2}+1.
        \end{aligned}
    \end{align}
    Therefore by Markov's inequality 
    \begin{equation}
    \label{Prob2_lemma1}
        \begin{split}
            \mathbb{P}_{\boldsymbol{\theta}}\{T_i > (1+\epsilon)\rho n_i^* \} &= \mathbb{P}_{\boldsymbol{\theta}}\left\{(\rho n_i^*)^{-1/2}(T_i-\rho n_i^*) > \epsilon (\rho n_i^*)^{1/2}\right\} \\
        &\leq \frac{\mathbb{E}_{\boldsymbol{\theta}}\bigg\{\Big |(\rho n_i^*)^{-1/2}(T_i-\rho n_i^*) \Big |^s\bigg\}}{\epsilon^s (\rho n_i^*)^{s/2}} = O(n_i^{*-s/2}), \quad \text{ if } m>s/2+1,\ s>0.
        \end{split}
    \end{equation}
    Combining \eqref{Prob1_lemma1} and \eqref{Prob2_lemma1} yields the desired result.
\end{proof}

\begin{proof}[\textbf{Proof of Theorem \ref{thm1}}]
Since $N_i \geq N_i'$ on $\{N_i' < T_i\}$ and $N_i = N_i'$ otherwise, it follows that $N_i \geq N_i'$ almost surely. Let $I_A$ denote the indicator function of an event $A$. Then,
\begin{align*}
\mathbb{E}_{\boldsymbol{\theta}}(N_i) - \mathbb{E}_{\boldsymbol{\theta}}(N_i')
&= \mathbb{E}_{\boldsymbol{\theta}}\left[(N_i - N_i') I_{\{N_i' < T_i\}}\right] \\
&\leq \mathbb{E}_{\boldsymbol{\theta}}\left[T_i \, I_{\{N_i' < T_i\}}\right] \\
&\leq \left[\mathbb{E}_{\boldsymbol{\theta}}(T_i^3)\right]^{1/3}
\left[\mathbb{P}_{\boldsymbol{\theta}}(N_i' < T_i)\right]^{2/3}
\quad \text{(by Hölder's inequality)}.
\end{align*}
Consequently, from \eqref{Ti_prop} and Lemma \ref{lemma1} we have 
\begin{align*}
    \mathbb{E}_{\boldsymbol{\theta}}(N_i) - \mathbb{E}_{\boldsymbol{\theta}}(N_i') &= O(n_i^*) \, O(n^{*-s/3})\\ 
    &= o(1), \quad \text{if } m > 5/2.
\end{align*}
The result in part (ii) now follows from the expression for $\mathbb{E}_{\boldsymbol{\theta}}(N_i')$ given in part (i) of \eqref{prop1_purely}, and part (i) follows immediately as a corollary.
\end{proof}

\begin{proof}[\textbf{Proof of Theorem \ref{thm2}}]
Using the same reasoning as in \eqref{indep_of_N=n_and_min_o.s.}, we conclude that
\[
I(\boldsymbol{N} = \boldsymbol{n}) \text{ is independent of } \left(X_{n_1(1)},\, Y_{n_2(1)}\right).
\]
Hence, the type I error probability can be written as
\begin{align*}
\mathbb{P}_{\boldsymbol{\theta},\boldsymbol{N}}\left(\text{Type I error}\right) 
&= \mathbb{E}_{\boldsymbol{\theta}}\left[\left(\frac{\sigma_1}{N_1}+\frac{\sigma_2}{N_2}\right)^{-1}
\frac{\sigma_1}{N_1}
\exp\left\{-\frac{N_1(b-\Delta_0)}{\sigma_1} \right\} \right] \\
&= \mathbb{E}_{\boldsymbol{\theta}}\left[\left(\frac{n_1^*}{N_1}+\frac{n_2^*}{N_2}\right)^{-1}
\frac{n_1^*}{N_1}
\exp\left\{-l_1\frac{N_1}{n_1^*} \right\} \right], \,\ \text{ where } l_1=d_\alpha (b-\Delta_0)
\end{align*}
Let $E$ denote the event
\[
E = \big[ (N_1' \geq T_1) \cap (N_2' < T_2)\big] 
\cup \big[ (N_1' < T_1) \cap (N_2' \geq T_2) \big] 
\cup \big[ (N_1' < T_1) \cap (N_2' < T_2) \big],
\]
for which $\mathbb{P}_{\boldsymbol{\theta}}(E) = O(q^{-s/2})$ if $m > s/2+1$. Now consider
\begin{align}
\label{diff_of_P(type-I error)}
\begin{aligned}
\big|\mathbb{P}_{\boldsymbol{\theta},\boldsymbol{N}}\left(\text{Type I error}\right)
&- \mathbb{P}_{\boldsymbol{\theta},\boldsymbol{N'}}\left(\text{Type I error}\right)\big| \\
&\leq \mathbb{E}_{\boldsymbol{\theta}}\Bigg[
\Bigg|\left(\frac{n_1^*}{N_1}+\frac{n_2^*}{N_2}\right)^{-1}
\frac{n_1^*}{N_1}
\exp\left\{-l_1\frac{N_1}{n_1^*} \right\}- \left(\frac{n_1^*}{N_1'}+\frac{n_2^*}{N_2'}\right)^{-1}
\frac{n_1^*}{N_1'}
\exp\left\{-l_1\frac{N_1'}{n_1^*} \right\}\Bigg| \, I_E \Bigg] \\
&\leq \mathbb{E}_{\boldsymbol{\theta}}\left[
\frac{n_1^*}{l_1}\left(\frac{1}{N_1} + \frac{1}{N_1'}\right) I_E
\right] \quad \text{(by \eqref{H_bound})} \\
&\leq O(n_1^*) \, \mathbb{P}_{\boldsymbol{\theta}}(E) = O(q)O(q^{*-s/2}) \\
&= o(1), \quad \text{if } m > 2.
\end{aligned}
\end{align}
Combining \eqref{diff_of_P(type-I error)} with \eqref{type-I_error_2nd_order_exp2} yields the desired result.
\end{proof}

\begin{proof}[\textbf{Proof of Theorem \ref{thm3}}]
    Proceeding in a similar fashion as in the proof of Theorem \ref{thm2}, the risk function can be expressed as 
    \begin{align*}    
    R_{\boldsymbol{N}}(\boldsymbol{\mu},\boldsymbol{c}) &= \mathbb{E}_{\boldsymbol{\theta}}\left[ A \left(\frac{\sigma_1}{N_1}+\frac{\sigma_2}{N_2}\right)^{-1}\frac{\sigma_2}{N_2}\exp\left\{-\frac{N_2(\Delta_1-b)}{\sigma_2} \right\} + c_1 N_1\sigma_1^{-1} + c_2 N_2 \sigma_2^{-1}\right] \\
    &= \mathbb{E}_{\boldsymbol{\theta}}\left[A\left(\frac{n_1^*}{N_1}+\frac{n_2^*}{N_2}\right)^{-1}\frac{n_2^*}{N_2}\exp\left\{-l_2\frac{N_2}{n_2^*} \right\}+c_1 N_1\sigma_1^{-1} + c_2 N_2 \sigma_2^{-1} \right], \,\ \text{ where } l_2=d_\alpha (\Delta_1-b).
    \end{align*}
Consider the difference
    \begin{align}
    \label{diff_of_Regret}
    \begin{aligned}    
    \big|\text{(Regret)}_{\boldsymbol{N}} - \text{(Regret)}_{\boldsymbol{N'}}
    \big| & =
    \big|R_{\boldsymbol{N}}(\boldsymbol{\mu},\boldsymbol{c})
    - R_{\boldsymbol{N'}}(\boldsymbol{\mu},\boldsymbol{c})\big| \\
    &\leq \mathbb{E}_{\boldsymbol{\theta}}\Bigg[
    \Bigg|A\left(\frac{n_1^*}{N_1}+\frac{n_2^*}{N_2}\right)^{-1}\frac{n_2^*}{N_2}\exp\left\{-l_2\frac{N_2}{n_2^*} \right\}- A\left(\frac{n_1^*}{N_1'}+\frac{n_2^*}{N_2'}\right)^{-1}\frac{n_2^*}{N_2'}\exp\left\{-l_2\frac{N_2'}{n_2^*} \right\}\\
    &\qquad +\frac{c_1}{\sigma_1} N_1 + \frac{c_2}{\sigma_2}N_2-\frac{c_1}{\sigma_1} N_1' - \frac{c_2}{\sigma_2}N_2'\Bigg| \, I_E \Bigg] \\
    &\leq \mathbb{E}_{\boldsymbol{\theta}}\left[\left\{\frac{n_2^*}{l_2}\left(\frac{1}{N_2} + \frac{1}{N_2'}\right) + \frac{c_1}{\sigma_1} (N_1+N_1') + \frac{c_2}{\sigma_2}(N_2+N_2')\right\}
     I_E \right] \\
    &\leq O(q) \, \mathbb{P}_{\boldsymbol{\theta}}(E) = o(1), \quad \text{if } m > 2.
    \end{aligned}
    \end{align}
Thus, part (ii) follows from \eqref{regret_2nd_order_exp2}, and part (i) is a direct consequence.
\end{proof}

\end{document}